\documentclass[10pt, pra, twocolumn]{revtex4}
\usepackage{hyperref}
\usepackage{verbatim}
\usepackage{amsmath}
\usepackage{latexsym}
\usepackage{revsymb}
\usepackage{yfonts}
\usepackage{soul}

\usepackage{natbib}
\usepackage{amsfonts}
\usepackage{amsmath}
\usepackage{amssymb}
\usepackage{amsthm}
\usepackage{graphicx}
\usepackage{bm}
\usepackage{bbm}

\newcommand{\be}{\begin{eqnarray} \begin{aligned}}
\newcommand{\ee}{\end{aligned} \end{eqnarray} }
\newcommand{\benn}{\begin{eqnarray*} \begin{aligned}}
\newcommand{\eenn}{\end{aligned} \end{eqnarray*} }

\newcommand{\bc}{\begin{center}}
\newcommand{\ec}{\end{center}}

\newcommand{\id}{\mathbb{I}}

\newcommand{\tr}{\mathop{\mathrm{tr}}\nolimits}

\newtheorem{theorem}{Theorem}[section]

\newtheorem{lemma}[theorem]{Lemma}

\newcommand{\finproof}{$\Box$}

\usepackage{amsfonts}
\def\Real{\mathbb{R}}
\def\Complex{\mathbb{C}}

\def\id{\mathbb{I}}

\def\01{\{0,1\}}

\newcommand{\ket}[1]{|#1\rangle}
\newcommand{\bra}[1]{\langle#1|}

\newcommand{\proj}[1]{|#1\rangle\langle#1|}

\newcommand{\inp}[2]{\langle{#1}|{#2}\rangle} 

\newcommand{\expect}[3]{\left\langle{#1}\right|{#2}\left|{#3}\right\rangle}

\newcommand{\mB}{\mathcal{B}}
\newcommand{\mC}{\mathcal{C}}

\newcommand{\mG}{\mathcal{G}}
\newcommand{\mL}{\mathcal{L}}
\newcommand{\mO}{\mathcal{O}}
\newcommand{\mK}{\mathcal{K}}

\newcommand{\setX}{\mathcal{X}}
\newcommand{\setA}{\mathcal{A}}
\newcommand{\setB}{\mathcal{B}}
\newcommand{\setS}{\mathcal{S}}
\newcommand{\assign}{:=}
\newcommand{\ox}{\otimes}

\newcommand{\steph}[1]{}
\newcommand{\prabha}[1]{}

\begin{document}

\title{A transform of complementary aspects with applications to entropic uncertainty relations}

\author{Prabha Mandayam}
\email[]{prabhamd@caltech.edu}
\affiliation{Institute for Quantum Information, California Institute of Technology, Pasadena CA 91125, USA}
\author{Niranjan Balachandran}
\email[]{nbalacha@caltech.edu}
\affiliation{Department of Mathematics, California Institute of Technology, Pasadena CA 91125, USA}
\author{Stephanie Wehner}
\email[]{wehner@caltech.edu}
\affiliation{Institute for Quantum Information, California Institute of Technology, Pasadena CA 91125, USA}
\date{\today}
\begin{abstract}
Even though mutually unbiased bases and entropic uncertainty relations play an important role in quantum cryptographic
protocols they remain ill understood.
Here, we construct special sets of up to $2n+1$ mutually unbiased bases (MUBs) in dimension $d=2^n$
which have particularly beautiful symmetry properties derived from the Clifford algebra. More precisely, we show 
that there exists a unitary transformation that cyclically permutes such bases. This unitary can be understood
as a generalization of the Fourier transform, which exchanges two MUBs, to multiple 
complementary aspects.
We proceed to prove a lower bound for min-entropic entropic uncertainty relations for any set of MUBs,
and show that symmetry plays a central role in obtaining tight bounds.
For example, we obtain for the first time a tight bound for four MUBs in dimension $d=4$, which
is attained by an eigenstate of our complementarity transform.
Finally, we discuss the relation to other symmetries obtained by transformations in discrete phase space, and
note that the extrema of discrete Wigner functions are directly related to 
min-entropic uncertainty relations for MUBs.
\end{abstract}
\maketitle

\section{Introduction}
One of the central ideas of quantum mechanics is the uncertainty principle which was first proposed by Heisenberg~\cite{heisenberg:ur} for
two conjugate observables. 
Indeed, it forms one of the most significant examples showing that quantum mechanics does 
differ fundamentally from the classical world. 
Uncertainty relations today are probably best known
in the form given by Robertson~\cite{robinson:uncertainty}, who extended Heisenberg's result to two arbitrary observables
$A$ and $B$. Robertson's relation states that if we prepare many copies of the state $\ket{\psi}$, and measure
each copy individually using either $A$ or $B$, we have
\begin{align}\label{eq:heisenberg}
\Delta A \Delta B \geq \frac{1}{2} |\bra{\psi}[A,B]\ket{\psi}|
\end{align}
where $\Delta X = \sqrt{\bra{\psi}X^2\ket{\psi} - \bra{\psi}X\ket{\psi}^2}$ for $X \in \{A,B\}$ is
the standard deviation resulting from measuring $\ket{\psi}$ with observable $X$.
The essence of~\eqref{eq:heisenberg} is that quantum mechanics does not allow us to simultaneously
specify definite outcomes for two non-commuting observables when measuring the same state.
The largest possible lower bound in Robertson's inequality (\ref{eq:heisenberg})
is $1$, 
which happens if and only if $A$ and $B$ are related by a Fourier transform, that is, they are conjugate observables.

Of particular importance to quantum cryptography is the case where $A$ and $B$ correspond to measurements in two 
different orthonormal bases $\setA = \{\ket{a}\}_a$ and $\setB = \{\ket{b}\}_b$ in dimension $d$. 
If $\setA$ and $\setB$ are related by the Fourier
transform, for all basis vectors $\ket{a}$ of basis $\setA$ and all vectors $\ket{b}$ of basis $\setB$,
\begin{align}
|\inp{a}{b}|^2 = \frac{1}{d}\ .
\end{align}
Any two bases satisfying this property are called \emph{mutually unbiased bases}, or \emph{complementary aspects}, and the unitary that exchanges
two mutually unbiased bases can be understood as a Fourier transform. 
In the light of Robertson's uncertainty relation~\eqref{eq:heisenberg},
it seems that bases which are related by the Fourier transform
should play a special role in our understanding of quantum mechanics, in the sense that
they are the measurements which are most ``incompatible''.

However, nature typically allows us to perform more than two measurements on any given system, leading to the natural question
of how we can determine ``incompatibility'' between multiple measurements. Clearly, due to its use of the commutator relation, the lower bound 
of~\eqref{eq:heisenberg} most directly relates to the case of \emph{two} measurements. 
Is there a natural way of quantifying uncertainty for multiple measurements? And if so, what measurements might be most ``incompatible''?

\subsection{Entropic uncertainty relations}

A natural measure that captures the relations among the probability distributions over the outcomes
for each observable is the entropy of such distributions. 
This prompted Hirschmann to propose the first entropic 
uncertainty relation for position and momentum observables~\cite{hirschmann:ur}.
This relation was later improved by~\cite{becker:ur,BB:M}, where~\cite{BB:M} show that Heisenberg's uncertainty relation~\eqref{eq:heisenberg}
is in fact implied by this entropic uncertainty relation. Hence, using entropic quantities provides us with a much more general way 
of quantifying uncertainty. 
Indeed, it was realized by Deutsch~\cite{deutsch:ur} 
that other means of quantifying ``uncertainty'' are 
also desirable for another reason:
Note that the lower bound in~\eqref{eq:heisenberg} is trivial when $\ket{\psi}$ happens to give
zero expectation on $[A,B]$. Hence, it would be useful to have a
way of measuring ``incompatibility'' which depends only on the measurements
$A$ and $B$ and not on the state.
Deutsch~\cite{deutsch:ur} himself showed that
\begin{align}\label{eq:deutsch}
\frac{1}{2}\left(H_\infty(\setA|\ket{\psi}) + H_\infty(\setB|\ket{\psi})\right) \geq - \log\left(\frac{1 + c(\setA,\setB)}
{2}\right)
\end{align}
where $c(\setA,\setB) \assign \max\{|\inp{a}{b}| \mid \ket{a} \in \setA, \ket{b} \in \setB\}$,
and 
\begin{align}\label{eq:minEntropy}
H_\infty(\setA|\ket{\psi}) = - \log \max_{a} |\inp{a}{\psi}|^2
\end{align}
is the min-entropy arising from measuring
the pure state $\ket{\psi}$ using the basis $\setA$ (see Section~\ref{sec:URquantities} for more information on the entropic quantities we use).
If $\setA$ and $\setB$ are related by a Fourier transform, then the r.h.s. of~\eqref{eq:deutsch} becomes 
$- \log(1/2 + 1/(2\sqrt{d}))$, where the minimum is achieved by a state that is invariant under the Fourier transform.
Since the Shannon entropy obeys $H(\cdot) \geq H_\infty(\cdot)$, Deutsch's bound also holds for the Shannon entropy.
Better lower bounds have since been obtained for the Shannon entropy by Maassen and Uffink~\cite{maassen:ur} following a conjecture
of Kraus~\cite{kraus:ur}. Their uncertainty relations are again strongest (in the sense that the lower bound is largest) when
the bases $\setA$ and $\setB$ are conjugate, that is, the two bases are related by a Fourier transform.
Apart from their role in understanding the foundations of quantum mechanics, these uncertainty relations play a central role
in cryptography in the noisy-storage model~\cite{prl:noisy,arxiv:robust,noisy:new, serge:bounded,serge:new}, quantum key
distribution~\cite{qkd:ur, joeRenes:ur},
information locking~\cite{barbara:locking}, and the question of separability~\cite{guehne:separable}. In particular, such relations have practical interest in noisy-storage cryptography~\cite{noisy:new} where they may enable us to prove security for a larger class quantum memories.

Here, we are concerned with measurements in multiple bases $\setB_0,\ldots,\setB_{L-1}$.
Entropic uncertainty relations provide a natural way of quantifying ``incompatibility'' of more than two measurements by lower bounding
\begin{align}\label{eq:H1general}
\frac{1}{L} \sum_{j=0}^{L-1} H(\setB_j|\ket{\psi}) \geq c_L\ ,
\end{align}
for all states $\ket{\psi}$. We call a state $\ket{\psi}$ that minimizes the average sum of entropies a \emph{maximally certain state}.
When $H$ is the Shannon entropy, the largest bound we can hope to obtain for any choice of bases is 
\begin{align}\label{eq:H1maxBound}
c_L = \frac{L-1}{L}\ \log d\ ,
\end{align}
since choosing $\ket{\psi}$ to be an element of one of the bases yields zero entropy when we subsequently measure in the same basis.
If~\eqref{eq:H1maxBound} is indeed a lower bound to~\eqref{eq:H1general}, we will call the measurements \emph{maximally incompatible with respect to the Shannon entropy}.
Note that this can only happen if $\ket{\psi}$ gives us full entropy~\footnote{Note that the entropy of performing a measurement corresponding to an orthonormal basis in dimension
$d$ can never exceed $\log d$, where the maximum is attained when the distribution over the outcomes is uniform $1/d$.} 
when measured in any other basis, that is, the bases are all mutually unbiased.

Curiously, however, it was shown that whereas being mutually unbiased is necessary, it is not a 
sufficient condition to obtain maximally strong uncertainty relations for the Shannon entropy~\cite{BallesterWehner}.
In particular, there do exist large sets of up to $\sqrt{d}$ mutually unbiased bases in square dimensions
for which we do obtain very weak uncertainty relations~\cite{BallesterWehner}. 
Recently, Ambainis~\cite{andris} has shown that 
for any three bases from the ``standard'' mutually unbiased bases construction~\cite{wootters:mub,boykin:mub} in
prime dimension, the lower bound cannot exceed $\left( \frac{1}{2}+o(1) \right) \log d$,
for large dimensions. For dimensions of the form $4k+3$ and $8k+5$ no further assumption is needed, but the proof 
assumes the Generalized Riemann Hypothesis for dimensions of the form $8k+1$.
Furthermore, for any $0\leq \epsilon\leq 1/2$, there always exist $k=d^\epsilon$ of these bases such
that the lower bound cannot be larger than $\left( \frac{1}{2}+\epsilon+o(1) \right)\log d$.
Only if we use the maximal set of $d+1$ mutually unbiased bases that can be found for any given prime power dimension, do we
obtain quite strong uncertainty relations~\cite{sanchez:old, ivanovic:ur}. 

At present, we merely know that there do exist arbitrarily large sets of two outcome measurements that give us maximally strong uncertainty relations~\cite{ww:cliffordUR}, and that in larger dimensions
selecting a large amount of bases at random does provide us with strong relations~\cite{andreas:random} (for a survey see~\cite{ww:urSurvey}).
Indeed, it remains an intriguing open question as to whether there even exist three measurements with three outcomes in dimension $d > 2$ that are maximally incompatible with respect to the Shannon entropy.

\subsection{Mutually unbiased bases}

In the light of these questions, it is therefore natural to study the structure of mutually unbiased bases(MUBs) to see whether
we can identify additional properties which are sufficient for obtaining strong uncertainty relations.
In~\cite{wootters:mur}, Wootters and Sussman made the interesting observation that for the maximal set of $d+1$ mutually unbiased bases coming
from such constructions as~\cite{wootters:mub, boykin:mub} in dimension $d=2^n$, the lower bound of the entropic uncertainty relation in terms of the collision
entropy given in~\cite{BallesterWehner} is tight, and the minimum is attained by a state that is invariant under a unitary that cyclically 
permutes the set of all $d+1$ MUBs. A similar unitary was noted to exist by Chau~\cite{chau:cycle}.
Wootters and Sussman derive their transformation from phase space arguments. Their unitary can in fact easily be generalized to cyclically permute $L$ bases, whenever $L$ divides $d+1$ (see Section~\ref{sec:URdwgGeneral}). The results in~\cite{wootters:mur} have recently been generalized by Appleby~\cite{appleby:cycle}, who shows that in prime power dimensions of the form $d= 1 \mbox{ or } 3 \mod 4$, there exists a unitary operation that cyclically permutes the first and second halves of the full set of MUBs. This raises the pressing question of whether smaller sets of MUBs also exhibit such symmetries? 
And can we exploit such symmetries to obtain tight uncertainty relations? In particular, is the minimizing state always an invariant of such a transformation 
as observed for two bases in~\eqref{eq:deutsch}\ ?

{\bf Main result}
We first show by an explicit construction that there exist sets of $2 \leq L \leq 2n+1$ mutually unbiased bases in dimension $d=2^n$ with the property that there exists a unitary that cyclically permutes all bases in this set, whenever (a) $L$ is prime, and (b) $L$ divides $n$ or $L=2n+1$. More specifically, we provide an explicit construction of MUBs $\setB_0,\ldots,\setB_{L-1}$ with $\setB_j = \{\ket{b^{(j)}}\}_b$ and a unitary
$U$ such that
\begin{align}\label{eq:basisPerm}
U \proj{b^{(j)}} U^\dagger &= \proj{b^{(j + 1 \mod L)}}\\
 &\mbox{ for all } \ket{b^{(j)}} \in \setB_j\ .\nonumber
\end{align}
Furthermore, in dimension $d=4$, we actually find such a unitary for any set of $L$ MUBs, where $2\leq L \leq 5$.
Our approach exploits properties of the Clifford algebra, which might yield new insights into the structure of these MUBs. 
It is entirely distinct from the phase space approach which was
used to construct such a unitary for the full set of $d+1$ MUBs~\cite{wootters:mub}. 
Note that our construction gives at most $O(\log d)$ bases, but 
shows that there is indeed an additional symmetry which has previously gone unnoticed.
For $L=2$ bases, $U$ is simply the Fourier transform, and it would be interesting to investigate
general properties of our transformation and whether it has applications in other areas.

\subsection{Min-entropic uncertainty relations}

We then apply our transformation to the study of uncertainty relations in terms of the \emph{min-entropy} (see~\eqref{eq:minEntropy}).
Since $H(\cdot) \geq H_\infty(\cdot)$, this also provides us with bounds on uncertainty relations in terms of the Shannon entropy.
Of course, many forms of entropy could be considered when it comes to quantifying uncertainty, and each has its merits. 
The min-entropy is 
of particular interest in cryptography, and is also related to the well studied extrema of the discrete Wigner function as we will discuss
in Section~\ref{sec:URdiscreteWigner}. In particular, it will be easy to see that the average min-entropy for the full set of
$L = d + 1$ MUBs can be bounded as
\begin{align}\label{eq:minEntropyWigner}
\frac{1}{d+1} \sum_{j=0}^{d} H_{\infty}(\setB_j|\ket{\psi}) \geq - \log \left[d \cdot \left(\max_{\alpha} W_\alpha^{\rm max} + 1\right)
\right]\ ,
\end{align}
where $W_\alpha^{\rm max}$ is the maximum value of the discrete Wigner function at the point $\alpha$ in discrete phase space.
Symmetries thereby play an important role in determining $W^{\rm max}_{\alpha}$.

{\bf Second result} We prove a simple min-entropic uncertainty relation for an arbitrary set of $L$ mutually unbiased bases.
For MUBs $\setB_0,\ldots,\setB_{L-1}$ we obtain
\begin{align}\label{eq:URstatement}
\frac{1}{L} \sum_{j=0}^{L-1} H_\infty(\setB_j|\ket{\psi}) 
\geq - \log \left[\frac{1}{L}\left(1 + \frac{L-1}{\sqrt{d}}\right)\right]. 
\end{align}
For the case of $2$ MUBs in dimension $d$, this bound is indeed the same as the bound in~\eqref{eq:deutsch}. For any small set of $2<L<d$ MUBs, our bound matches the strongest known bound~\cite{chris:diss}. We also prove the following alternate lower bound:
\begin{align}\label{eq:URstatement2}
\frac{1}{L} \sum_{j=0}^{L-1} H_\infty(\setB_j|\ket{\psi}) 
\geq - \log \left[\frac{1}{d}\left(1 + \frac{d-1}{\sqrt{L}}\right)\right],
\end{align}
which is stronger than~\eqref{eq:URstatement} for the complete set of $L = d+1$ MUBs in dimension $d$. Clearly, when $L=d$, the two bounds are equivalent. 

We further show that~\eqref{eq:URstatement} is in fact tight for $L=4$ MUBs in dimension $d=4$ stemming from our construction, where the minimum is 
attained for an invariant state of the transformation $U$ that cyclically permutes all $4$ bases. Even though this is a somewhat restricted
statement, it is the first time that a tight entropic uncertainty relation has been obtained for this case. The minimizing
state here has an appealing symmetry property, just as for the case of $L=2$ bases in~\eqref{eq:deutsch} where the minimum
is attained by a state that is invariant under the Fourier transform. 

For the collision entropy $H_2$, Wootters~\cite{wootters:mur} previously showed that the lower bound from~\cite{BallesterWehner} is attained
by an invariant state when considering the full set of $d+1$ MUBs. Here, however, we 
exhibit a tight uncertainty relation for these $L=3$ bases in $d=4$ for the collision entropy $H_2$ which has an entirely different
structure and the minimum is not attained by an invariant state of our transformation. 
Nevertheless, we have for the first time a \emph{tight} entropic
uncertainty relation for \emph{all} possible MUBs in a dimension larger than the trivial case of $d=2$ where the Bloch sphere
representation makes the problem easily accessible. In $d=4$, we have a tight relation for $H_\infty$ for $L=2,4$, and tight
relations for $H_2$ for $L=3,5$. 

Our result indicates that due to the different properties of the minimizing 
state for different numbers of bases, the problem may be even more daunting than previously imagined. Yet, our work shows that
in each case the minimizing state is by no means arbitrary. It has a well defined (albeit different) structure in each of the cases.

{\bf Third result} For some sets of MUBs we do obtain for the first time, significant insight into the structure of the maximally certain states.
In particular, we note in Section~\ref{sec:symmetry} that for $L$ mutually unbiased bases that the state that minimizes the min-entropic uncertainty relations is
an invariant of a certain unitary whenever $L$ divides $d+1$ for $d=2^n$. 

\section{Symmetric MUBs}\label{sec:defs}

Before explaining our construction of mutually unbiased bases for which there exists a unitary that cyclically permutes them, 
let us define the notions of MUBs more formally and recall some known facts.
Let $\setB_1 = \{\ket{0^{(1)}},\ldots,\ket{d-1^{(1)}}\}$ and $\setB_2 =
\{\ket{0^{(2)}},\ldots,\ket{d-1^{(2)}}\}$ be two orthonormal bases in
$\Complex^d$. They are said to be
\emph{mutually unbiased} if
$|\inp{a^{(1)}}{b^{(2)}}| = 1/\sqrt{d}$, for all $a,b \in \{0,\ldots,d-1\}$. 
A set $\{\setB_0,\ldots,\setB_{L-1}\}$ of
orthonormal bases in $\Complex^d$ is called a \emph{set of mutually
unbiased bases} if each pair of bases is mutually unbiased.
For example, the well-known computational and Hadamard basis are mutually unbiased.
We use $N(d)$ to denote the maximal number of MUBs in dimension $d$.
In any dimension $d$, we have that
$N(d) \leq d+1$~\cite{boykin:mub}. If $d = p^k$ is a prime power, we have $N(d) = d+1$ and explicit constructions are
known~\cite{boykin:mub,wootters:mub}. Other constructions are known that give
less than $d+1$ MUBs in other dimensions~\cite{wocjan:mub,zauner:diss,klappenecker:mubs, grassl:mub}. 
However, it is still an open problem whether there exists a set of $7$ 
(or even $4$!) MUBs in dimension $d=6$.

\subsection{Clifford algebra}

Our construction of mutually unbiased bases makes essential use of the techniques developed in~\cite{boykin:mub},
together with properties of the Clifford algebra.
The Clifford algebra is the associative algebra generated by operators $\Gamma_0,\ldots,\Gamma_{2n-1}$ satisfying
$\{\Gamma_i,\Gamma_j\} = 0 \mbox{ for } i\neq j$
and $\Gamma_i^2 = \id$.
It has a unique representation by
Hermitian matrices on $n$ qubits (up to unitary equivalence) that can
be obtained via the famous Jordan-Wigner transformation~\cite{JordanWigner}:
\begin{align}
  \Gamma_{2j+1} &= Y^{\ox(j-1)} \ox Z \ox \id^{\ox(n-j)}, \\
  \Gamma_{2j}   &= Y^{\ox(j-1)} \ox X \ox \id^{\ox(n-j)},
\end{align}
for $j=0,\ldots,n-1$, where we use $X$, $Y$ and $Z$ to denote the Pauli matrices.
Furthermore, we let 
\begin{align}
\Gamma_{2n} := i \Gamma_0\ldots\Gamma_{2n-1}\ .
\end{align}
Note that in dimension $d=2$ these are just the familiar Pauli matrices, $\Gamma_0 = X$, $\Gamma_1 = Z$ and 
$\Gamma_2 = Y$.

Of particular importance to us will be the fact that we can view the operators $\Gamma_0,\ldots,\Gamma_{2n-1}$, as
$2n$ orthogonal vectors forming a basis for $\Real^{2n}$. In particular, for any orthonormal transformation
$T \in \text{O}(2n)$ which when applied to the vector $v =(v^{(0)},\ldots,v^{(2n-1)})\in \Real^{2n}$ 
gives $\tilde{v} = (\tilde{v}^{(1)}, \ldots,\tilde{v}^{(2n-1)}) = T(v)$, there exists a unitary 
$U(T)$ such that
\begin{align}
U(T)\left(\sum_j v_j \Gamma_j\right) U(T)^{\dagger} = \sum_j \tilde{v}_j \Gamma_j\ .
\end{align}
The orthonormal transformation that is particularly interesting to us here is the one that cyclically permutes
the basis vectors.
As described above we can find a corresponding unitary $U = U(T)$ which 
cyclically permutes the basis vectors $\Gamma_0,\Gamma_2,\ldots, \Gamma_{L-1}$.
An explicit construction can be found in the appendix.
This symmetry can be extended to $\text{SO}(2n+1)$, see e.g.~\cite{ww:cliffordUR}. 
It will also be useful that the set of $d^{2}$ operators 
\begin{align}
\setS = \{\id,\Gamma_j, i \Gamma_i \Gamma_j, \Gamma_i\Gamma_j\Gamma_k,\ldots,
i \Gamma_1\ldots\Gamma_{2n}\} 
\end{align}
forms an orthogonal basis~\footnote{Orthogonal with respect to the Hilbert-Schmidt inner product.} for $d\times d$ Hermitian matrices in $d=2^n$~\cite{dietz:blochsphere}.

\subsection{Construction}
To construct mutually unbiased bases, we follow the procedure outlined in~\cite{boykin:mub}, but now applied
to a subset of the operators in $\setS\setminus\{\id\}$.
That is, we will 
group operators into classes of commuting operators, i.e., sets $\{\mathcal{C}_{0},\mathcal{C}_{1},\ldots,\mathcal{C}_{L-1} 
\mid \mathcal{C}_j \subset \setS\setminus\{\id\}\}$ of size $|\mathcal{C}_j| = d-1$
such that
\begin{enumerate}
\item[(i)] the elements of $\mathcal{C}_{j}$ commute for all $0 \leq j \leq L-1$,
\item [(ii)] $\mathcal{C}_{j}\cap\mathcal{C}_{k} = \emptyset$ for all $j\neq k$. 
\end{enumerate}
It has been shown in~\cite{boykin:mub} that the common eigenbases of such classes form a set of $L$ MUBs.

First of all, note that no class can contain two generators $\Gamma_j$ and $\Gamma_k$ since they do not commute. 
When forming the 
classes we hence ensure that 
each one contains exactly one generator $\Gamma_j$, 
which clearly limits us to constructing at most $2n+1$ such classes.
The difficulty in obtaining a partitioning that is suitable for our purpose is to ensure that the 
unitary $U$ that cyclically permutes the generators $\Gamma_0,\ldots,\Gamma_{L-1}$ also permutes the corresponding
bases by permuting products of operators appropriately.
We show in the appendix that our general construction achieves the following:
\begin{theorem}
Suppose that $2 \leq L \leq 2n+1$ is prime, and either $L$ divides $n$ or $L=2n+1$.
Then in dimension $d = 2^n$, there exist $L$ mutually unbiased bases $\mB_0,\ldots,\mB_{L-1}$ for which there exists a unitary $U$ that
cyclically permutes them
\begin{align}\label{eq:Udefine}
U \mB_j = \mB_{j + 1 \mod L}\ .
\end{align}
\end{theorem}

\subsection{Examples}

Let us consider two simple examples of such classes in dimension $d=4$. These are not obtained from our general construction, but nevertheless
provide us with the necessary intuition. 
For $L=3$ MUBs the classes are given by
\begin{eqnarray}\label{eq:classd3}
\mathcal{C}_{0} &=& \{\Gamma_{0}, i\Gamma_{1}\Gamma_{4}, i\Gamma_{3}\Gamma_{2} \} \nonumber \\
\mathcal{C}_{1} &=& \{\Gamma_{1}, i\Gamma_{2}\Gamma_{4}, i\Gamma_{3}\Gamma_{0} \} \nonumber \\
\mathcal{C}_{2} &=& \{\Gamma_{2}, i\Gamma_{0}\Gamma_{4}, i\Gamma_{3}\Gamma_{1} \} 
\end{eqnarray}
It is easy to see that the unitary $U$ that achieves the transformation 
$\Gamma_{0} \rightarrow \Gamma_{1} \rightarrow \Gamma_{2} \rightarrow \Gamma_{0}$,
but leaves $\Gamma_3$ and $\Gamma_{4}$ invariant, cyclically permutes the bases given above.
For the collision entropy $H_2$ the minimum is attained for an eigenstate of the \emph{commuting}
operators $\Gamma_0$, $i\Gamma_2 \Gamma_{4}$ and $i \Gamma_3 \Gamma_1$. 
This minimizing state also shows that the wellknown bound for $H_2$ (see e.g.~\cite{ww:urSurvey}) can be \emph{tight}. 

For $L=4$ MUBs we obtain the classes
\begin{eqnarray}\label{eq:classd4}
\mathcal{C}_{0} &=& \{\Gamma_{0}, i\Gamma_{1}\Gamma_{4}, i\Gamma_{2}\Gamma_{3} \} \nonumber \\
\mathcal{C}_{1} &=& \{\Gamma_{1}, i\Gamma_{2}\Gamma_{4}, i\Gamma_{3}\Gamma_{0} \} \nonumber \\
\mathcal{C}_{2} &=& \{\Gamma_{2}, i\Gamma_{3}\Gamma_{4}, i\Gamma_{0}\Gamma_{1} \} \nonumber \\
\mathcal{C}_{3} &=& \{\Gamma_{3}, i\Gamma_{0}\Gamma_{4}, i\Gamma_{1}\Gamma_{2} \}
\end{eqnarray}
It is easy to see that the unitary $U$ that achieves the transformation 
$\Gamma_{0} \rightarrow \Gamma_{1} \rightarrow \Gamma_{2} \rightarrow \Gamma_{3} \rightarrow \Gamma_{0}$, 
but leaves $\Gamma_4$ invariant, cyclically permutes the bases given above.
For $L=4$ classes the minimum in the entropic uncertainty relation for $H_\infty$ is attained
for a state that is invariant under the transformation $U$.
However, we also know that for $L=4$ or $L=8$ classes in dimension $d=8$ no partitioning of operators is possible
that satisfies our requirements. Thus, for what values of $L$ and $d$ such a unitary can be found, remains an interesting open question.

\section{Uncertainty relations}\label{sec:URrelations}

We now investigate the relationship between the observed symmetries and entropic uncertainty relations.

\subsection{Entropic quantities}\label{sec:URquantities}

Before comparing different uncertainty relations, we first provide a short introduction to all the entropic quantities we will use.
The expert reader may safely skip this section. In general, 
the \emph{R{\'e}nyi entropy}~\cite{renyi:entropy} of order $\alpha$ of the distribution obtained by measuring a state $\ket{\psi}$
in the basis $\setB = \{\ket{b}\}_b$ is given by
\begin{align}\label{eq:RenyiGeneral}
H_\alpha(\setB|\ket{\psi}) = \frac{1}{1-\alpha}\log\left[\left(\sum_{b \in \setB}(|\inp{b}{\psi}|^2)^\alpha\right)^\frac{1}{\alpha-1}\right].
\end{align}
Indeed, the Shannon entropy forms a special case of the R{\'e}nyi entropy by taking the limit $\alpha \rightarrow 1$, i.e.,
$H_1(\cdot) = H(\cdot)$, where we omit the subscript.
Of particular importance are the \emph{min-entropy},
for $\alpha \rightarrow \infty$:\label{def:minentropy}
\begin{align}
H_\infty(\setB|\ket{\psi}) = - \log\left(\max_{b \in \setB} |\inp{b}{\psi}|^2\right)\ ,
\end{align}
and the \emph{collision entropy}\label{def:collision}
\begin{align}
H_2(\setB|\ket{\psi}) = - \log \sum_{b \in \setA} (|\inp{b}{\psi}|^2)^2\ .
\end{align}
We have
\begin{align}
\log d\geq H(\cdot) \geq H_2(\cdot) \geq H_\infty(\cdot) \geq 0\ ,
\end{align}
and hence uncertainty relations for $H_\alpha$ also provide us with a bound on uncertainty relations for $H_\beta$ whenever
$\alpha \geq \beta$.

Note that intuitively, the min-entropy is determined by the highest peak in the distribution and most closely captures 
the notion of ``guessing''. To see why it is a more useful quantity in cryptography than the Shannon entropy, consider the following
example distribution $P_X$:
Let $\setX = \01^n$ and let $x_0 = 0,\ldots,0$ be the all 0 string. Suppose that $P_X(x_0) = 1/2 + 1/(2^{n+1})$
and $P_X(x) = 1/(2^{n+1})$ for $x \neq x_0$, i.e., with probability $1/2$ 
we choose $x_0$ and with probability $1/2$ we choose one string
uniformly at random. Then $H(X) \approx n/2$, whereas $H_\infty(X) \approx 1$! If $x$ would correspond to an encryption
key used to encrypt an $n$ bit message, 
we would certainly not talk about security if we can guess the key with probability at least $1/2$ ! Yet, the Shannon entropy is quite high.

\subsection{Min-entropy and symmetry}\label{sec:URmin-entropy}

Apart from its cryptographic applications, min-entropic uncertainty relations are appealing since 
the problem of determining tight uncertainty relations can be simplified considerably in the presence of symmetries.
Furthermore, these relations bear an interesting relation to the extrema of the discrete Wigner function.
First of all, note that for the min-entropy we have by Jensen's inequality that
\begin{align}\label{eq:minEntropyBound}
&\frac{1}{L} \sum_{j=0}^{L-1} H_\infty(\mB_j|\rho) \\
 &\qquad\geq - \log \frac{1}{L} \sum_{j=0}^{L-1} \max_{b^{(j)}} \tr(\rho \proj{b^{(j)}})\, 
\label{eq:minEntropyBoundLower}
\end{align}
where the inequality becomes equality if all terms $\tr(\rho \proj{b^{(j)}})$ are the same. For $\vec{b} = (b^{(0)},\ldots,b^{(L-1)}) \in 
\{0,\ldots,d-1\}^{\times L}$, define
\begin{align}\label{eq:pVecs}
P_{\vec{b}} := \sum_{b^{(j)}} \proj{b^{(j)}}\ .
\end{align}
Note that determining a tight lower bound to~\eqref{eq:minEntropyBoundLower} is thus equivalent to determining
\begin{align}
\max_{\vec{b}} \max_\rho \tr(\rho P_{\vec{b}})\ .
\end{align}
Clearly, any $\zeta$ such that 
\begin{align}\label{eq:zetaBound}
P_{\vec{b}} \leq \zeta \id \mbox{ for all } \vec{b}
\end{align}
thus gives us a lower bound for~\eqref{eq:minEntropyBound}.
For any set of bases, this makes the problem of finding a bound more approachable as it reduces the problem to finding the largest eigenvalue for any
operator $P_{\vec{b}}$. In particular, it can be phrased as a semidefinite program to minimize $\zeta$ such that~\eqref{eq:zetaBound}
holds for all $\vec{b}$.

\subsubsection{Symmetries}\label{sec:symmetry}

It is now easy to see why symmetries simplify our goal 
of determining tight uncertainty relations.

\begin{lemma}\label{lem:symmetry}
Suppose that for every $\vec{b} \in \{0,\ldots,d-1\}$ there exists
a unitary $U_{\vec{b}}$ such that $U_{\vec{b}} \ket{b^{(j)}} = \ket{b^{(j + 1 \mod L)}}$. Then
there exists a $\vec{b'}$ such that
the minimum in~\eqref{eq:minEntropyBound} is attained for a state $\rho$ that is invariant
under $U_{\vec{b'}}$.
\end{lemma}
\begin{proof}
First of all, note that 
\begin{align}
\frac{1}{L} \sum_{j=0}^{L-1} (U_{\vec{b}}^j) P_{\vec{b}} (U_{\vec{b}}^j)^{\dagger} = P_{\vec{b}}\ ,
\end{align}
and hence for $\rho_{\rm sym} = (1/L) \sum_j (U_{\vec{b}}^j)^\dagger \rho (U_{\vec{b}}^j)$
\begin{align}
\tr(\rho_{\rm sym} P_{\vec{b}}) = \tr(\rho P_{\vec{b}})\ .
\end{align}
In particular, this holds for the state $\rho = \proj{\psi}$ corresponding
to the eigenvector $\ket{\psi}$ with the largest eigenvalue of $P_{\vec{b'}}$. 
When looking for the minimizing state on the r.h.s of~\eqref{eq:minEntropyBound} we can thus restrict ourselves to states which 
are invariant under $U_{\vec{b'}}^j$.
Note that in this case, we furthermore have that
\begin{align}
\tr(\rho_{\rm sym} \proj{b^{(j)}}) &= \frac{1}{L} \tr(\rho P_{\vec{b}})\ ,
\end{align}
meaning that the inequality~\eqref{eq:minEntropyBound} is tight in case of such a symmetry which is our claim.
\end{proof}

The question of course remains, whether such unitaries do exist in general.
Wootters and Sussman~\cite{wootters:mub} 
have shown that there exists a unitary $U$ that cyclically permutes the set of all $d+1$ MUBs for $d=2^n$ by constructing
a unitary that corresponds to a rotation around the origin in phase space. Clearly, by considering the unitary $U^k$ one can trivially adapt
their construction to obtain a unitary that cyclically permutes $L$ MUBs whenever $L \cdot k = d+1$. By first translating any point in the phase space to the origin, then applying the transformation $U^k$ and finally translating the origin back to the original
point, one can obtain the desired unitaries $U_{\vec{b}}$ that enable us to find tight bounds for the min-entropic uncertainty relations. 
This is the first time we gain significant insight into the structure of the states that minimize~\eqref{eq:minEntropyBound}.

Note that our construction only gives unitaries $U_{\vec{b}}$ for $\vec{b} = (c,\ldots,c)$ for any $c \in \{0,\ldots,d-1\}$. This means
that our complementarity transform $U$, leads to tight bounds only if the largest eigenvalue of any $P_{\vec{b}}$ happens to occur for a $\vec{b}$ 
of this form. This is for example the case for $L=4$ in $d=4$, where we do not obtain a unitary from the phase space approach 
of~\cite{wootters:mub}.

\subsubsection{Discrete Wigner function}\label{sec:URdiscreteWigner}\label{sec:URdwgGeneral}
To see how finding a lower bound for min-entropic uncertainty relations
for $d+1$ MUBs relates to finding the extrema of the discrete Wigner function, let us first recall the properties of the discrete Wigner function. The discrete phase space is a two-dimensional vector space over a finite field $\mathbb{F}_d$, where here we focus on the case of
$d=2^n$. For every state $\rho$, we can associate a function $W_\alpha$ with every point $\alpha$ in the discrete phase space, 
known as the discrete Wigner function. For completeness, we provide a short summary on how to determine $W_\alpha$; a detailed
account can be found in~\cite{gibbons:striations}.
First of all, note that the $d^2$ points of the discrete phase space can be parititioned into 
$d$ parallel lines each of which contains $d$ points. Any such partition is called a 
\emph{striation}, and it is known that $d+1$ such striations can be found~\cite{gibbons:striations}.
One may now define the discrete Wigner function by relating each striation to one of the $d+1$ possible
mutually unbiased bases~\cite{gibbons:striations}: Let $\lambda_{b,j}$ denote the $b$-th line in the striation $j$. With each such line, 
we associate a projector
\begin{align}
Q(\lambda_{b,j}) = \proj{b^{(j)}}\ ,
\end{align}
onto the $b$-th element of the basis $\mB_j$, in a specific order so as to satisfy certain symmetry constraints~\cite{gibbons:striations}.
Defining the phase-space point operator 
\begin{align}
A_\alpha := \sum_{\substack{\lambda_{b,j}\\\alpha \subset \lambda_{b,j}}} Q(\lambda_{b,j}) - \id\ ,
\end{align}
one can now define the discrete Wigner function as
\begin{align}\label{eq:wignerFunction}
W_\alpha := \frac{1}{d} \tr(A_\alpha \rho)\ .
\end{align}
The \emph{extrema of the discrete Wigner function} are defined
as the minimum and maximum of~\eqref{eq:wignerFunction} over quantum states $\rho$. 

Note that when considering $L=d+1$ mutually unbiased bases, 
each point $\alpha$ in the discrete phase space can be contained in exactly one line from each basis, as all lines in a striation, i.e.,
one basis are parallel. Hence, there is a one-to-one correspondence between points $\alpha$ in discrete phase space and vectors
$\vec{b} \in \{0,\ldots,d-1\}^{\times d+1}$. In terms of the phase space operator this means that $A_\alpha + \id = P_{\vec{b}}$
Note that the maximum of the discrete Wigner function
\begin{align}
W_\alpha^{\rm max} = \max_\rho \frac{1}{d} \tr(A_\alpha \rho)\ ,
\end{align}
is simply the largest eigenvalue of $A_\alpha$ (or $P_{\vec{b}} - \id$) up to a factor of $1/d$. 
We thus have that 
\begin{align}
\zeta := d \cdot \left[\max_{\alpha} W_\alpha^{\rm max} + 1\right]\ ,
\end{align}
satisfies $P_{\vec{b}} \leq \zeta \id$ and the maximum of the discrete Wigner function provides 
a lower bound to the min-entropic uncertainty relations as given in~\eqref{eq:minEntropyWigner}.
The extrema $W_\alpha^{\rm max}$ were evaluated numerically in~\cite{casaccino:discrete} for small $d$. Note, however,
that as noted in Section~\ref{sec:symmetry}, one may use symmetries to solve the problem of determining
$W_\alpha^{\rm max}$ directly.

\subsection{A simple bound}\label{sec:URsimpleBound}

As mentioned in Section~\ref{sec:URmin-entropy}, the problem of finding a lower bound for the average min-entropy reduces to the problem of finding the maximum eigenvalue of the operator $P_{\vec{b}}$ defined in \eqref{eq:pVecs}. In the appendix, we use a result due to Schaffner~\cite{chris:diss} obtained using the techniques of Kittaneh~\cite{kittaneh3:normsum}, to show that for any set of $L$ mutually unbiased bases in dimension $d$, the maximum eigenvalue of $P_{\vec{b}}$ is bounded by
\begin{equation}
P_{\vec{b}} \leq \frac{1}{L}\left(1 + \frac{L-1}{\sqrt{d}}\right)\id, \mbox{ for all } \vec{b}.
\end{equation}
Using this, we obtain the following simple bound for the average min-entropy
in the appendix
\begin{lemma}\label{lem:smallLBound}
Let $\mB_0,\ldots,\mB_{L-1}$ be a set of mutually unbiased bases in dimension $d$.
Then,
\begin{equation}\label{eq:newbound}
\frac{1}{L}\sum_{j=0}^{L-1}\mathcal{H}_{\infty}(\mathcal{B}_{j}|\ket{\psi}) \geq -\log\left[\frac{1}{L}\left(1 + \frac{L-1}{\sqrt{d}}\right)\right]\ . 
\end{equation}
\end{lemma}
For the case of $L=2$ MUBs in dimension $d$, our bound exactly matches the well known result of Deutsch (see ~\eqref{eq:deutsch}). For $L>2$, the only other known lower bound for the average min-entropy is the one obtained in ~\cite{chris:diss}, where it is shown that for a set of $L<\sqrt{d}$ MUBs in dimension $d=2^{n}$, the following holds: 
\begin{eqnarray}\label{eq:chrisbound}
\hspace{-1mm}&&\frac{1}{L}\sum_{j=0}^{L-1}\mathcal{H}_{\infty}(\mathcal{B}_{j}|\ket{\psi}) \nonumber \\
\hspace{-1mm}&\geq& \hspace{-1mm} -\log\left[\frac{1}{L}\left(1 + \frac{L-1}{\sqrt{d}}\max_{0\leq i<j\leq L-1}\sqrt{|X^{i}||X^{j}|}\right)\right]
\end{eqnarray}
where $X^{i}, X^{j} \subset \{0,1\}^{n}$ are subsets of $n$-bit strings. In the case of min-entropic uncertainty relations, these subsets contain only a single string, which corresponds to the peak of the probability distribution induced on the state $\ket{\psi}$ by the corresponding bases $\mathcal{B}^{i}$ and $\mathcal{B}^{j}$, so that,
\begin{align}
\max_{0\leq i<j\leq L-1}\sqrt{|X^{i}||Y^{j}|} = 1,
\end{align}
Thus, in dimension $d=2^{n}$, our bound in \eqref{eq:newbound} is clearly the same as \eqref{eq:chrisbound}. The reason we obtain a more general bound for $L \leq d+1$ MUBs in any dimension $d$ is that we reduce the problem directly to an eigenvalue problem without going through the representation in terms of bit strings as in~\cite{chris:diss}, via a much simpler argument using only the concavity of the $\log$.

Using an alternate approach involving a Bloch sphere like representation of the basis vectors $\ket{b^{(j)}}$, we show that the maximum eigenvalue of $P_{\vec{b}}$ can be bound differently, as follows:
\begin{equation}
P_{\vec{b}} \leq \frac{1}{d}\left(1 + \frac{d-1}{\sqrt{L}}\right)\id, \mbox{ for all } \vec{b}.
\end{equation}
As we show in the appendix, this implies
\begin{lemma}\label{lem:largeLBound}
Let $\mB_0,\ldots,\mB_{L-1}$ be a set of mutually unbiased bases in dimension $d$.
Then,
\begin{equation}\label{eq:(d+1)bound}
\frac{1}{L}\sum_{j=0}^{L-1}\mathcal{H}_{\infty}(\mathcal{B}_{j}|\ket{\psi}) \geq -\log\left[\frac{1}{d}\left(1 + \frac{d-1}{\sqrt{L}}\right)\right] \ .
\end{equation}
\end{lemma}
Notice that this alternate bound on the min-entropy is stronger than \eqref{eq:newbound} when $L>d$. In particular, for the complete set of $d+1$ MUBs in dimension $d$, this alternate bound in \eqref{eq:(d+1)bound} is stronger than any previously known bounds. When $L=d$ the two bounds that we derive are indeed equivalent. See Figs.~\ref{fig:dim4} and~\ref{fig:dim8} for a comparison of our bounds in dimensions $d=4$ and $d=8$ respectively.

\begin{figure}[!htp]
\begin{center}
	\includegraphics[width=0.52\textwidth]{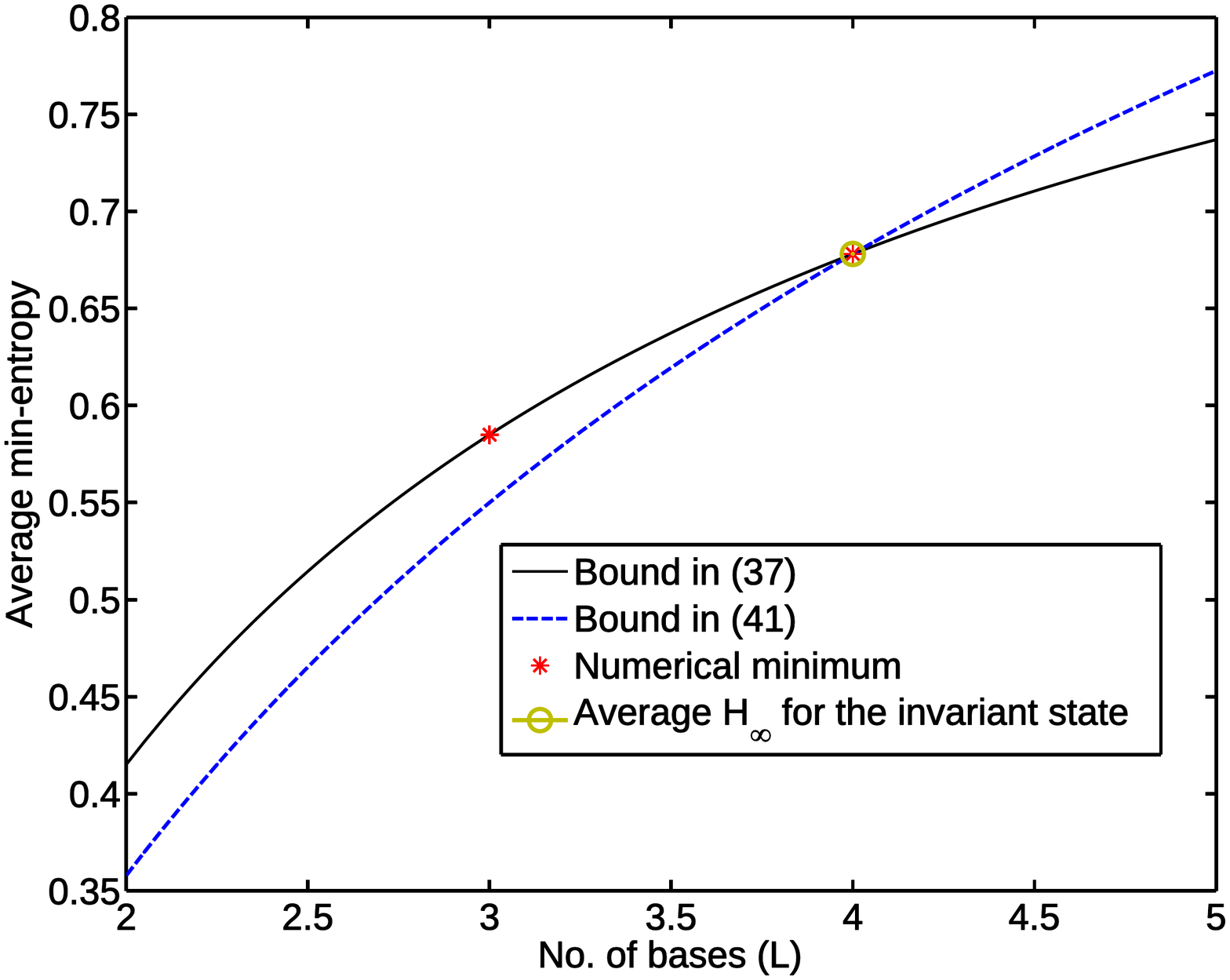}
	\caption{\label{fig:dim4} Average min-entropy for different sets of MUBs in dimension $d=4$. \\
The crosses denote numerically computed minima of the average min-entropy for MUBs obtained using our construction. 
The bound in \eqref{eq:newbound} is clearly tight for both $L=3$ and $L=4$ MUBs. 
The second analytical bound in \eqref{eq:(d+1)bound} is stronger than \eqref{eq:newbound} for $L = d+1 = 5$ bases. The circle denotes the average min-entropy for the invariant states given in \eqref{eq:invstate}. For $4$ MUBs in $d=4$ the minimum of the average min-entropy is indeed attained by states invariant under $U$.}
\end{center}
\end{figure}

Finally, we provide an example of a set of MUBs where Lemma~\ref{lem:smallLBound} is tight. For the set of $L=4$ MUBs in dimension $d=4$ constructed from the classes given in \eqref{eq:classd4}, our bound
\begin{align}
&\frac{1}{4}\sum_{j=0}^{3}\mathcal{H}_{\infty}(\mathcal{B}_{j}|\ket{\psi}) \geq -\log\left[\frac{1}{4}\left(1 + \frac{3}{2}\right)\right]
\approx 0.678\ ,
\end{align}
is tight, and the minimum is indeed achieved by a state that is invariant under the unitary transform that cycles through the bases, as defined in \eqref{eq:Udefine}. As noted in Section~\ref{sec:symmetry}, in this case, the largest eigenvalue of $P_{\vec{b}}$ occurs for a $\vec{b}$ of the form $\vec{b} = (c,\ldots,c)$ for any $c \in \{0,\ldots,3\}$. The states that achieve the lower bound are in fact eigenvectors of $U$, which can be expressed in terms of the MUB basis vectors as follows,
\begin{equation}\label{eq:invstate}
\ket{\psi_{b}} = \frac{1}{2} \sum_{j=0}^{3}\exp(i\pi j/4)\ket{b^{(j)}}, \; b \in \{0,\ldots,3\}.
\end{equation}

\begin{figure}[!ht]
  \begin{center}
	\includegraphics[width=0.52\textwidth]{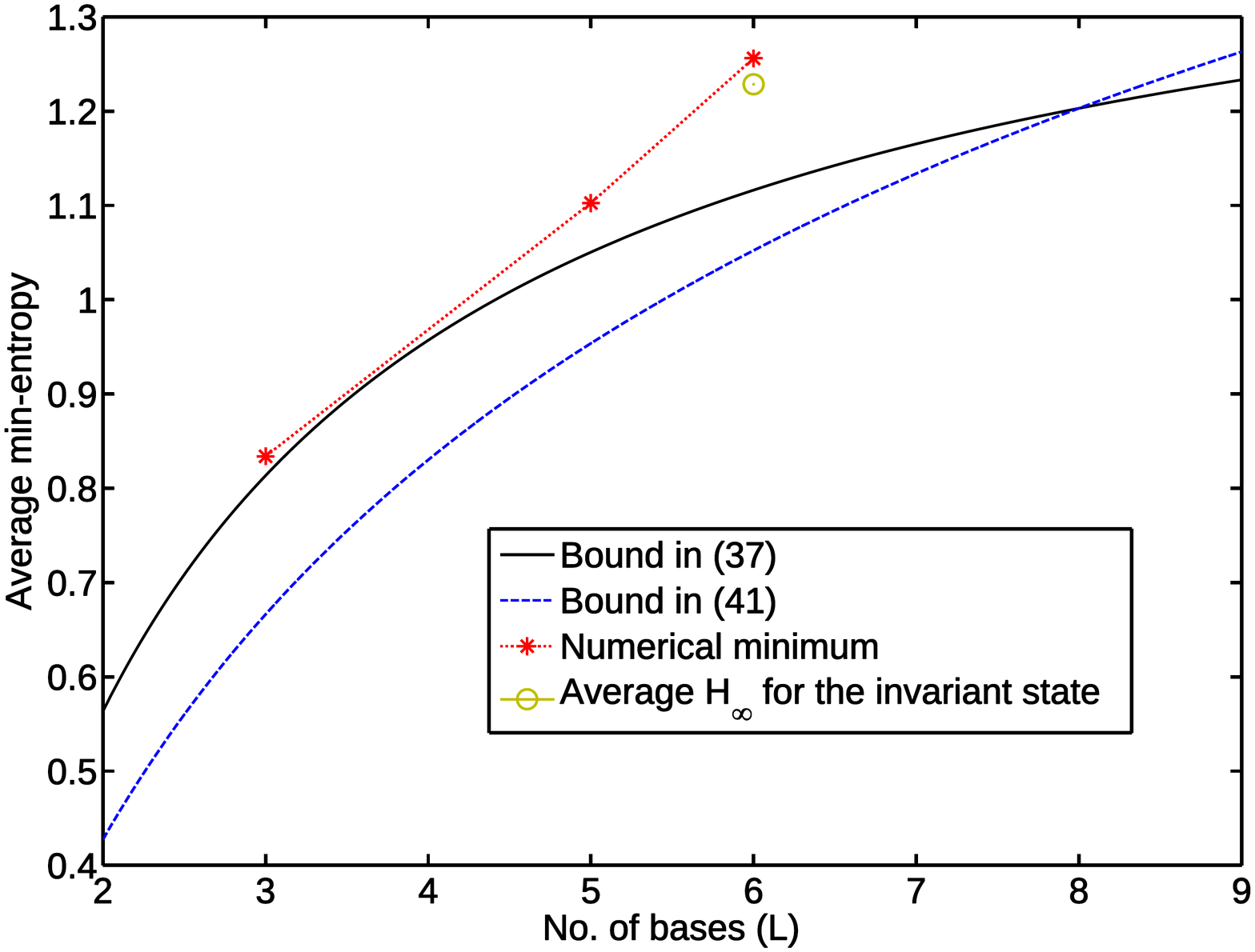}
	\caption{\label{fig:dim8} Average min-entropy for different sets of MUBs in dimension $d=8$.\\ 	
The bound in \eqref{eq:newbound} is close to tight for $L=3$ MUBs in dimension $d=8$. 
The second analytical bound in \eqref{eq:(d+1)bound} is stronger than \eqref{eq:newbound} for $L = d+1 = 9$ bases. The circle denotes the average min-entropy for invariant states constructed in dimension $d=8$, similar to the states described in \eqref{eq:invstate}. For $6$ MUBS in $d=8$, the minimum of the average min-entropy is nearly attained by states invariant under $U$.}
  \end{center}  
\end{figure}

\section{Conclusions and open questions}

We have shown that there exist up to $2 \leq L \leq 2n+1$ mutually unbiased bases in dimension $d=2^n$ for which we can find a unitary that cyclically permutes these bases, whenever $L$ is 
prime and $L$ divides $n$ or $L=2n+1$. This unitary is found by exploiting symmetry
properties of the Clifford algebra. Our approach is quite distinct from the phase space approaches that were previously used to show
that there exists such a unitary for the set of all $d+1$ MUBs~\cite{wootters:mub}, or for two halves of the full sets of MUBs when
$d= 1 \mbox{ or } 3 \mod 4$~\cite{appleby:cycle}. Our unitary can be understood as a generalization of the Fourier transform, and it would be interesting
to see whether it has other applications in quantum information. 

It is an interesting open question to generalize our result to other dimensions, or to a different number of bases. In prime dimension,
one could consider the generalized Clifford algebra~\cite{algebra:book}. 
Even though it does not have the full ${\rm SO}(2n+1)$ symmetry, it nevertheless exhibits enough symmetries to allow an exchange of generators.
This stems from the way the (generalized) Clifford algebra is obtained~\cite{algebra:book,genClifford}, which 
permits any transformation that preserves the $p$-norm for $p \geq 2$ in dimension $p$.
Yet, this is only the first step of our construction. As for generalizing our result to any $L$ bases in dimension $d=2^n$, we note
that it \emph{is} indeed possible to find such classes even when $L$ is not prime, as our example for $L=4$ in dimension $d=4$ shows. However, we also know that for $L=8$ classes in dimension $d=16$, no partitioning of operators can be found satisfying our requirements. It is an interesting open question as to when such a partitioning can be found in general.

Finally, we use our complementarity transform to obtain a tight uncertainty relation for the min-entropy for $L=4$ bases in dimension $d=4$.
No tight relations are known for this case before. We also use a slight generalization of the unitary from~\cite{wootters:mub} to show that when $d=2^n$ and $L$ divides $d+1$, the
minimizing state is an invariant of a certain unitary. This is the first time that significant insight has been obtained on the structure of the minimizing states for min-entropic uncertainty relations for mutually unbiased bases. It is an exciting open question to obtain tight relations in general, and understand the structure of the minimizing states.

\acknowledgments

We are grateful to David Gross for pointing us to the relevant literature for the discrete phase space construction. 
We also thank Lukasz Fidkowski and John Preskill for interesting discussions. PM and SW are supported by NSF grant PHY-0803371.

\appendix

\bigskip

\section{Constructing the unitary $U$}

It is well known that for any orthonormal transformation $T \in \text{O}(2n)$ there exists a corresponding unitary transformation
$U(T)$~\cite{lounesto:book}, where we refer to~\cite[Appendix C]{steph:diss} for instructions on how to 
obtain explicit constructions.
The transformation we wish to construct here, of the form $U(T)\Gamma_j U(T)^\dagger \rightarrow \Gamma_k$, is thereby particularly simple to obtain. It can be build up from successive rotations
in the plane spanned by only two ``vectors'' $\Gamma_j$ and $\Gamma_k$.
More specifically, we first construct a unitary that corresponds to a rotation around an angle $\pi/2$ in the plane spanned
by $\Gamma_j$ and $\Gamma_k$, bringing $\Gamma_j$ to $\Gamma_k$. This is simply a reflection around the plane 
orthogonal to the midvector between $\Gamma_j$ and $\Gamma_k$, followed by a reflection around the plane orthogonal
to $\Gamma_k$. Using the geometric properties of the Clifford algebra this corresponds to the unitary
\begin{align}
R_{j \rightarrow k} = \Gamma_k (\Gamma_j + \Gamma_k)/\sqrt{2}\ .
\end{align}
To obtain the desired unitary, we now compose a number of such rotations.
Let $\hat{R}_{j,k} = R_{j \rightarrow k}$ if $k$ is odd, and $\hat{R}_{j,k} = R_{k \rightarrow j}$ if $k$ is even.
Furthermore, let $F = \id$ if $L$ is odd, and $F = \Gamma_{2n} \Gamma_{L-1}$ if $L$ is even. Note that $\Gamma_{2n} \Gamma_{L-1}$
is the unitary that flips the sign of $\Gamma_{L-1}$, but leaves all $\Gamma_j$ for $j \neq 2n$ and $j \neq (L-1)$ invariant.
We may then write
\begin{align}
U(T) = F\hat{R}_{0,1}\hat{R}_{0,2}\ldots\hat{R}_{0,L-1}\ .
\end{align}

This unitary hence transforms $\Gamma_0 \rightarrow \Gamma_1 \rightarrow \ldots \rightarrow \Gamma_{L -1}
\rightarrow \Gamma_0$, but leaves all other generators $\Gamma_j$ for $j \geq L$ invariant. A similar unitary can be found for any transformation 
$T \in \text{SO}(2n+1)$~\cite{ww:cliffordUR}, but is more difficult to construct explicitly.

\section{Constructing maximally commuting classes of Clifford generators}

In ~\eqref{eq:classd3} and ~\eqref{eq:classd4} we gave examples of constructing $L=3$ and $L=4$ MUBs in dimension $d=4$, such that they are cyclically permuted under the action of a unitary $U$ that permutes the Clifford generators in $d=4$. Here, we show by a general construction that it is always possible to construct $L$ such classes in dimension $d=2^{n}$, whenever $L|n$ and $L$ is prime. We also outline a construction for $L=2n+1$ classes, given a unitary $U$ that cycles through \textit{all} $2n+1$ Clifford generators, when $2n+1$ is prime.

Given the $2n$ generators of the Clifford algebra in dimension $d=2^{n}$, we consider the set
\begin{align}
\setS = \{\id, \Gamma_{j}, i\Gamma_{j}\Gamma_{k}, \Gamma_{j}\Gamma_{k}\Gamma_{l},..., i\Gamma_{0}\Gamma_{1}..\Gamma_{2n-1}\equiv\Gamma_{2n}\}.
\end{align}
To generate a set of $L\leq 2n+1$ MUBs, we seek to group the elements of $\setS$ into $L$ classes of commuting operators, ie. sets  $\{\mathcal{C}_{0},\mathcal{C}_{1},\ldots,\mathcal{C}_{L-1} \mid \mathcal{C}_j \subset \setS\setminus\{\id\}\}$ of size $|\mathcal{C}_j| = d-1$
such that
\begin{enumerate}
\item[(\textbf{P1})] The elements of $\mathcal{C}_{j}$ commute for all $0 \leq j \leq L-1$,
\item[(\textbf{P2})] The classs are all \textit{mutually disjoint}, that is, 
\begin{align}
\mathcal{C}_{j}\cap\mathcal{C}_{k} =  \emptyset \; \mbox{ for all }\; j\neq k,
\end{align}
\item[(\textbf{P3})] The unitary $U$ that cyclically permutes the generators $\Gamma_0,\ldots,\Gamma_{L-1}$, also permutes the corresponding classes by permuting products of operators appropriately. 
\end{enumerate}
Our approach in obtaining such a set of classes is to first pick $d-1$ elements for the class $\mC_0$ and then generate the rest of the classes by repeated application of $U$ to the elements of $\mC_0$. This automatically ensures property~(\textbf{P3}). To ensure (\textbf{P1}) and (\textbf{P2}), the $d-1$ operators $\mC_0 \equiv \{\mO_1,\mO_2,\ldots\mO_{d-1}\}$ must satisfy the following:
\begin{enumerate}
\item[(i)] For any pair $\mO_i,\mO_j \in \mC_0$, $[\mO_i,\mO_j] = 0$, and
\item[(ii)] The operators in $\mC_0$ cycle through \textit{mutually disjoint} sets of operators under the action of $U$.
\end{enumerate} 
To understand condition (ii) better, consider an operator $\mO_i$ in $\mC_0$. Then, by construction, $U^{k}(\mO_i) \in \mC_{k}$ for $0 \leq k \leq L-1$, assuming we construct a total of $L$ classes. In addition, property (ii) implies $U^{k}(\mO_i) \notin \mC_j$, for any $j \neq k$. In other words, given any two operators $\mO_i,\mO_j \in \mC_0$ that cycle through the sets 
\begin{eqnarray}
S_i &=& \{U^{k}(\mO_i) | 0 \leq k \leq L-1 \} \; \textrm{and} \\ 
S_j &=& \{U^{k}(\mO_j) | 0 \leq k \leq L-1\}\ , 
\end{eqnarray}
respectively under the action of $U$, property(ii) demands that $S_i\cap S_j = \emptyset$, for all $i\neq j$ and $i,j = 1,2,\ldots,d-1$.

Finally, we note that no class can contain two generators $\Gamma_j$ and $\Gamma_k$, since they do not commute. When forming the classes we hence ensure that each one contains exactly one generator $\Gamma_j$, which we refer to as the \textit{singleton} $\Gamma$-operator of the class, as opposed to the rest of the elements which will be \textit{products} of $\Gamma$-operators. The fact that each class can contain at most one singleton operator limits us to constructing a maximum of $2n+1$ such classes. 

\subsection{Mathematical tools}

Before proceeding to outline our construction, we establish some useful mathematical facts which will help motivate our algorithm for the construction of mutually disjoint classes. For the rest of the section, we will work with a set of $p$ $\Gamma$-operators $\{\Gamma_0, \Gamma_2,\ldots,\Gamma_{p-1}\}$ that are cycled under the action of $U$, as follows, 
\begin{equation}
U: \Gamma_{0}\rightarrow\Gamma_{1}\rightarrow\ldots\Gamma_{p-1}\rightarrow\Gamma_{0}, 
\end{equation}
In other words, we are given a set of $\Gamma$-operators whose \textit{cycle-length} is $p$.

\subsection{Length-$2$ operators}

First, we consider sets of products of two $\Gamma$-operators of the form $\Gamma_{i}\Gamma_{j}$, which we call \textit{length-2} operators. It is convenient to characterize such pairs in terms of the \textit{spacing} -- ($S$) -- between the operators that constitute them. The spacing function $S$, for a given set of $p$ operators, is simply defined as: $S(\Gamma_{i}\Gamma_{j}) = (j-i)\,\textrm{mod} \, p$. Then, the following holds:
 
\begin{lemma}[\textbf{Unique spacings imply non-intersecting cycles}]\label{lem:spacing}
The action of $U$ on any length-2 operator $\Gamma_{i}\Gamma_{j}$ leaves its spacing function $S(.)$ invariant. Thus, length-2 operators that have unique spacings cycle through mutually disjoint sets of operators under the action of $U$.
\end{lemma}
\begin{proof}
Recall, $U: \Gamma_{i} \rightarrow \Gamma_{(i+1)\textrm{mod}p}$. It clearly follows that 
\begin{eqnarray}
U: S(\Gamma_{i}\Gamma_{j}) &\rightarrow& S\left(\Gamma_{(i+1)\textrm{mod}p}\Gamma_{(j+1)\textrm{mod}p}\right) \nonumber \\
&=& (j-i)\,\textrm{mod}\,p  \nonumber \\
&=& S(\Gamma_{i}\Gamma_{j}).
\end{eqnarray}
\end{proof}

\subsection{Higher length operators}

Similar to defining length-2 operators, we refer to any product of $\ell$ $\Gamma$-operators as a \textit{length-$\ell$} operator. For operators of length higher than $2$, it becomes convenient to refer to them using their corresponding index sets. For example, the operator $\Gamma_{i_1}\Gamma_{i_2}\ldots\Gamma_{i_\ell}$ will be simply denoted by the index set $(i_1,i_2,\ldots i_\ell)$. In the following Lemma, we obtain a condition for any set of length-$\ell$ operators to cycle through mutually disjoint sets under the action of $U$.

\begin{lemma}[\textbf{Mutually disjoint cycles for length $\ell$}]\label{lem:lengthl}
Suppose the length-$\ell$ operators (for $3\leq\ell\leq p-1$) that belong to the class $\mC_{0}$ 
are such that they correspond to index sets $(i_1,i_2,\ldots,i_{\ell})$ which all sum to the same value
\begin{equation}\label{eq:sumC0}
i_{1} + i_{2} + \ldots + i_{\ell} = c_{\ell}\,\textrm{mod}\,p, \; \forall \; (i_1,i_2,\ldots,i_\ell) \in \mC_0
\end{equation}
Then, no given index set of length $\ell$
can belong to more than one class, for prime values of $p$.
\end{lemma}
\begin{proof}
Given the operators $\{\Gamma_{i_{1}}\Gamma_{i_{2}}\ldots\Gamma_{i_{\ell}}\} \in \mC_{0}$, such that the corresponding index sets $(i_{1},i_{2},\ldots,i_{\ell})$ sum to 
\begin{align}
i_{1} + i_{2} + \ldots + i_{\ell} = c_{\ell}\,\textrm{mod}\,p, \; \forall (i_{1},i_{2},\ldots,i_{\ell}) \in \mC_0.
\end{align}
Under the action of $U$, these index sets change to
\begin{eqnarray}
(i_{1}, i_{2}, \ldots, i_{\ell}) &\rightarrow& (i^{(1)}_{1},i^{(1)}_{2},\ldots,i^{(1)}_{\ell}) \\
 &=& (i_{1}+1,i_{2}+1,\ldots,i_{\ell}+1)\,\textrm{mod}\,p\ . \nonumber
\end{eqnarray}
For any index set 
$(i^{(1)}_{1},i^{(1)}_{2},\ldots,i^{(1)}_{\ell})\in \mC_1$
the sum of the indices corresponding to the new operators $\{\Gamma_{i^{(1)}_{1}}\Gamma_{i^{(1)}_{2}}\ldots\Gamma_{i^{(1)}_{\ell}}\} \in \mC_{1}$ becomes
\begin{eqnarray}
i_{1}^{(1)} + i^{(1)}_{2} + \ldots + i^{(1)}_{\ell} &=& (c_{\ell} + \ell)\,\textrm{mod}\,p, 
\end{eqnarray}
Proceeding similarly, the corresponding operators in the class $\mC_{k}$ have index sets $(i_{1}^{(k)},i_{2}^{(k)},\ldots,i_{\ell}^{(k)})$ that sum to
\begin{eqnarray}\label{eq:suml}
i_{1}^{(k)} + i_{2}^{(k)} + \ldots + i_{\ell}^{(k)} &=& (c_{\ell} + k\ \ell)\,\textrm{mod}\,p,
\end{eqnarray}
for all $(i^{(k)}_{1},i^{(k)}_{2},\ldots,i^{(k)}_{\ell}) \in \mC_k$.
Thus, starting with a constraint on the length-$\ell$ operators in $\mC_0$, we have obtained a constraint on the corresponding operators in a generic class $\mC_k$.

Now, to arrive at a contradiction, suppose that an index set $(j_1,j_2,\ldots,j_\ell)$ whose indices $\{j_m\}_m$ take values from the set $\{0,1,\ldots,p-1\}$, belongs to two different classes, $\mC_k$ and $\mC_k'$ (with $k\neq k'$). The constraint imposed by~\eqref{eq:suml} implies
\begin{eqnarray}\label{eq:constraintl}
(c_{\ell} + k\ \ell)\,\textrm{mod}\,p &=& (c_{\ell} + k'\ \ell)\,\textrm{mod}\,p \nonumber \\
\Rightarrow (k-k')\ell\,\textrm{mod}\,p &=& 0 .
\end{eqnarray}
Without loss of generality, let $k> k'$. Since we can form at most $p$ classes, the difference $(k-k')$ can be at most $(p-1)$. 
Finally, since $\ell \leq p-1$, condition \eqref{eq:constraintl} cannot be satisfied for prime values of $p$. 
\end{proof}

Recall that our approach to constructing any $p$ classes is to first construct the class $\mC_0$, and then obtain the rest by successive application of $U$. Therefore, the fact that any index set of a certain length $\ell$ cannot belong to more than one class implies that each length-$\ell$ operator in $\mC_0$ cycles through a unique set of length-$\ell$ operators under $U$. In other words, the length-$\ell$ operators cycle through mutually disjoint sets, as desired.

Lemma~\ref{lem:lengthl} thus provides us with a sufficient condition for the set of length-$\ell$ operators in $\mC_0$ to cycle through mutually disjoint sets under $U$, given a set of $\Gamma$-operators whose cycle-length is prime-valued. We only need to ensure that the length-$\ell$ operators in the first class that we construct, $\mC_0$, correspond to index sets that \textit{all} sum to the same value. This condition is of course subject to the constraint that the maximum allowed length for the operators in $\mC_0$ (and by extension, in any class) is $p-1$. 

\subsection{Constructing $2n+1$ prime classes}

As a warmup, we construct $L=2n+1$ classes in dimension $d = 2^{n}$, when $2n+1$ is prime. 
This case is particularly easy, and illustrates how the results of the previous sections will be used in general.

\begin{theorem}[\textbf{$2n+1$ prime classes}]\label{thm:(2n+1)classes}
Let $\mG^{(\textrm{full})} = \{\Gamma_0,\ldots,\Gamma_{2n}\}$ denote the complete set of
$(2n+1)$ $\Gamma$-operators, and
let $U$ be the unitary that cycles through all of them, that is,
\begin{eqnarray}
U \; &:&  \Gamma_{0}\rightarrow\Gamma_{1}\ldots\Gamma_{2n-1}\rightarrow\Gamma_{2n}\rightarrow\Gamma_{0}\ .
\end{eqnarray}
If $2n+1$ is prime, then there exist $2n+1$ classes $\mC_0,\mC_1,\ldots\mC_{2n}$ satisfying properties \textbf{(P1)} through \textbf{(P3)}.
\end{theorem}

\begin{proof}
We prove the existence of $2n+1$ classes by construction. We first outline an algorithm to pick $d-1$ operators that constitute the class $\mC_{0}$. The remaining classes are easily obtained by the application of $U$ to the elements of $\mC_{0}$. Then, we make use of Lemmas~\ref{lem:spacing} and~\ref{lem:lengthl} to prove that the classes obtained through our construction do satisfy the desired properties.\\

\bigskip

\noindent\textbf{\underline{Algorithm}}
\begin{enumerate}
\item[1.] Pick one of the elements of $\mG^{(\textrm{full})}$, $\Gamma_{0}$, as the singleton operator.
\item[2.] Pair up the remaining operators in $\mG^{(\textrm{full})}$ to form $(n-1)$ length-2 operators which commute with $\Gamma_{0}$, as follows,
\begin{eqnarray}
\mL_{2} &=& \{ \,\Gamma_{1}\Gamma_{2n}, \; \Gamma_{2}\Gamma_{2n-1}, \; \ldots, \\
\; \; \; \; && \ldots , \Gamma_{n-2}\Gamma_{n+3}, \; \Gamma_{n-1}\Gamma_{n+2} \, \}, \nonumber
\end{eqnarray}
where $\mL_2$ denotes the set of length-2 operators in $\mC_1$. Since we have left out the pair $\Gamma_n \Gamma_{n+1}$ in the middle, we get,
$|\mL_{2}| = n-1$.
\item[3.] Form higher length operators that commute with $\mL_{2}\cup\{\Gamma_{0}\}$, by combining $\Gamma_{0}$ with appropriate combinations of the length-2 operators. Any operator of even length $\ell = 2j$ is created by combining $i$ pairs in $\mL_{2}$. And any operator of odd length $\ell = 2j+1$ is created by appending $\Gamma_{0}$ to a length-$2j$ operator.  

Denoting the sets of length-3 operators as $\mL_{3}$, length-4 operators as $\mL_{4}$, and in general, the set of length-$i$ operators as $\mL_i$, we have,
\begin{eqnarray}
|\mL_{3}| &=& |\mL_{2}| = n-1 ,\nonumber \\
|\mL_{4}| &=& \left(\begin{array}{c}
 n-1\\
 2\end{array}\right), \; |\mL_{5}| = |\mL_{4}|, \nonumber \\
|\mL_{6}| &=& \left(\begin{array}{c}
 n-1\\
 3\end{array}\right), \; |\mL_{7}| = |\mL_{6}|, \nonumber \\
 \vdots && \vdots \nonumber \\
|\mL_{2n-2}| &=& \left(\begin{array}{c}
 n-1\\
 n-1\end{array}\right) = 1,\; |\mL_{2n-1}| = |\mL_{2n-2}|. \nonumber
\end{eqnarray}
\end{enumerate}

Putting together the operators from steps $(1)$, $(2)$, and $(3)$ we get the desired cardinality for the class $\mC_{0}$ as follows:
\begin{eqnarray}\label{eq:countC0}
|\mC_{0}| &=& 1 + (n-1) + \sum_{i=3}^{2n}|\mL_{i}| \nonumber \\
&=& 1 + 2(n-1) + 2\left(\begin{array}{c}
 n-1\\
 2\end{array}\right) + 2\left(\begin{array}{c}
 n-1\\
 3\end{array}\right) \nonumber \\
 && \; \; \; \; + ... 
 + 2\left(\begin{array}{c}
 n-1\\
 n-1\end{array}\right) \nonumber \\
 &=& 2\sum_{i=0}^{n-1}\left(\begin{array}{c}
 n-1\\
 i\end{array}\right) -1 = 2(2^{n-1}) -1 \nonumber \\
 &=& 2^{n}-1 = d-1
\end{eqnarray}

The rest of the classes are generated by successive applications of the unitary $U$ to the elements of $\mC_{0}$, so that $U: \mC_{i}\rightarrow\mC_{(i+1)\mod 2n+1}$. \\

\bigskip

It is easy to see that the elements of each class satisfy property (\textbf{P1}) above -- the different length operators have been picked in such a way as to ensure that they all commute with each other. Similarly, by construction, they satisfy property \textbf{(P3)}. It only remains to prove property (\textbf{P2}), that the classes are all mutually disjoint.\\

The elements of $\mL_2$ correspond to the following set of spacings
\begin{align}
S(\mL_{2}) \equiv \{2n-1,2n-3, \ldots,5, 3\} 
\end{align}
which are all distinct. So by Lemma~\ref{lem:spacing}, the elements of $\mL_2$ cycle through mutually disjoint sets of length-$2$ operators. 

For higher length operators, we first show that our construction meets the conditions of Lemma~\ref{lem:lengthl}. For the class $\mC_0$, the elements of $\mL_2$ correspond to index sets that satisfy 
\begin{align}
\mL_{2}(\mC_{0}) = \{(i_{1},i_2)| \, i_1 + i_2 = 0\,\textrm{mod}\,(2n+1)\}.
\end{align}
The length-2 operators of a generic class $\mC_{k}$ similarly satisfy
\begin{equation}\label{eq:L2}
\mL_2(\mC_k) = \{(i_{1},i_{2})| \, i_1 + i_2 = 2k\,\textrm{mod}\,(2n+1)\}.
\end{equation}

Since higher length operators are essentially combinations of length-2 operators and the singleton operator, conditions similar to~\eqref{eq:L2} hold for higher length index sets as well. Since operators of even length $\ell = 2j$ contain $j$ pairs from $\mL_2$, the corresponding index sets in $\mC_0$ satisfy
\begin{eqnarray}
i_1+i_2+\ldots+i_{2j} &=& 0\,\textrm{mod}\,(2n+1), \nonumber \\
&\forall& \hspace{-1mm} (i_1,i_2,\ldots,i_{2j}) \in \mC_0.
\end{eqnarray}

Similarly, since the odd length operators have $\Gamma_0$ appended to the even length operators, the index sets of length $\ell = 2j+1$ in $\mC_0$ satisfy,
\begin{eqnarray}
i_{1}+i_2+\ldots+i_{2j+1} &=& 0\,\textrm{mod}\,(2n+1), \nonumber \\
&\forall& \hspace{-1mm} (i_1,i_2,\ldots,i_{2j+1}) \in \mC_0.
\end{eqnarray}

To sum up, for any $3\leq \ell\leq 2n$, our construction ensures that index sets of length-$\ell$ belonging to $\mC_0$ sum to the same value. The conditions of Lemma~\ref{lem:lengthl} are therefore satisfied, with the quantity $c_{\ell}$ in \eqref{eq:sumC0} taking the value $c_{\ell} =0$, for all $\ell = 3,\ldots,2n$. Now, we can simply evoke Lemma~\ref{lem:lengthl} to prove that, when $2n+1$ is prime, the higher length operators in $\mC_0$ cycle through mutually disjoint sets of operators.
\end{proof}

\subsection{Constructing $L|n$ classes for prime values of $L$}

Next, we show that it is possible to obtain an arrangement of operators into $L$ classes in dimension $2^{n}$, when $L$ is prime and $L|n$, such that the unitary $U$ that cyclically permutes $L$ $\Gamma$-operators
also permutes the coresponding classes.

\begin{theorem}[\textbf{$L|n$ classes for prime $L$}]\label{Lclasses}
Suppose $U$ is a unitary that cycles through sets of $L$ $\Gamma$-operators from the set $\mG^{\textrm{(full)}}\setminus\{\Gamma_{2n}\}$ in dimension $2^{n}$, where $L$ is prime and $L|n$. Then there exist $L$ classes $\mC_{0},\mC_1,\ldots\mC_{L-1}$ that satisfy properties \textbf{(P1)} through \textbf{(P3)}.
\end{theorem}

\noindent\textbf{Proof:} Note that since $L|n$ we have $n = rL$ for some positive integer $r$. The set of $2n$ Clifford generators $\Gamma_{0},\Gamma_{1},\ldots,\Gamma_{2n-1}$ can then be partitioned into $2r$ sets as follows:
\begin{eqnarray}
\mG^{(0)} &=& \{\Gamma_{0}, \Gamma_{1}, \ldots, \Gamma_{L-1}\}, \nonumber \\
\mG^{(1)} &=& \{\Gamma_{L}, \Gamma_{L+1},\ldots, \Gamma_{2L-1}\}, \nonumber \\
\vdots && \vdots \nonumber \\
\mG^{(2r-1)} &=& \{\Gamma_{(2r-1)L}, \Gamma_{(2r-1)L+1},\ldots,\Gamma_{2n-1}\}
\end{eqnarray}
Without loss of generality, we can assume the unitary $U$ is constructed such that it cyclically permutes the $L$ operators within each set, as follows.
\begin{eqnarray}
U &:& \Gamma_0\rightarrow\Gamma_1\rightarrow\ldots\rightarrow\Gamma_{L-1}\rightarrow\Gamma_0 , \nonumber \\
&& \Gamma_{L}\rightarrow\ldots\rightarrow\Gamma_{2L-1}\rightarrow\Gamma_{L} , \nonumber \\
&& \vdots \nonumber \\
&& \Gamma_{(2r-1)L}\rightarrow\ldots\rightarrow\Gamma_{2n-1}\rightarrow\Gamma_{(2r-1)L}. \nonumber
\end{eqnarray}

Once again, we begin with an algorithm for picking $d-1$ elements for the class $\mC_0$. The algorithm closely follows the one outlined in the previous section, barring some minor modifications. \\

\bigskip

\noindent \textbf{\underline{Algorithm}}
\begin{enumerate}
\item[1.] The ``middle'' element from $\mG^{(1)}$, $\Gamma_{(L-1)/2}$, is picked as the singleton element of $\mC_0$.
\item[2.] The $(n-1)$ length-2 operators which commute with $\Gamma_{(L-1)/2}$ are picked as follows -- 
\item[(a)] $\frac{L-3}{2}$ pairs are picked from $\mG^{(0)}\setminus\{\Gamma_{(L-1)/2}\}$
\begin{equation}
\mL_{2}^{(0)} = \{ \,\Gamma_{1}\Gamma_{L-1}, \, \Gamma_{2}\Gamma_{L-2},\ldots,\Gamma_{(L-3)/2}\Gamma_{(L+3)/2} \, \}, \nonumber
\end{equation}
leaving $\Gamma_0$ and $\Gamma_{(L+1)/2}$ unused.
\item[(b)] $\frac{L-1}{2}$ pairs are picked from each of the sets $\mG^{(1)}$ through $\mG^{(2r-1)}$,
\begin{eqnarray}
\mL_{2}^{(1)} &=& \{ \,\Gamma_{L+1}\Gamma_{2L-1}, \; \Gamma_{L+2}\Gamma_{2L-2}, \; \ldots , \nonumber \\
 && \; \; \; \ldots , \; \Gamma_{L+(L-1)/2}\Gamma_{L+(L+1)/2} \, \}, \nonumber \\
\vdots && \vdots \nonumber \\
\mL_{2}^{(2r-1)} &=& \{ \,\Gamma_{(2r-1)L+1}\Gamma_{2n-1}, \; \Gamma_{(2r-1)L+2}\Gamma_{2n-2}, \; \ldots , \nonumber \\
 &&   \; \Gamma_{(2r-1)L+(L-1)/2}\Gamma_{(2r-1)L+(L+1)/2} \, \}, \nonumber 
\end{eqnarray}
leaving the first operator in each set unused. 
\item[(c)] Finally, the unused $\Gamma$-operators from different sets are put together as specified below, to get the remaining $r$ length-$2$ operators:
\begin{equation}
\mL_{2}^{(2r)} = \{\,\Gamma_{0}\Gamma_{L}, \;\Gamma_{2L}\Gamma_{3L}, \ldots,\Gamma_{(2r-2)L}\Gamma_{(2r-1)L}\,\}. \nonumber 
\end{equation}
The set of length-$2$ operators is then given by 
\[
\mL_{2} = \mL_{2}^{(0)}\cup\mL_{2}^{(1)}\ldots\cup\mL_{2}^{(2r-1)}\cup\mL_{2}^{(2r)}
\]
which gives $|\mL_{2}| = \frac{L-3}{2} + (2r-1)\left(\frac{L-1}{2}\right) + r = rL - \frac{2r-2}{2} + r = n-1$.
\item[3.] Pick higher length operators from $\setS$ that commute with $\Gamma_{(L-1)/2}$ and $\mL_{2}$, by combining $\Gamma_{(L-1)/2}$ with appropriate combinations of the length-2 operators. As before, any even-length operator of length $\ell = 2i$ is obtained by combining $i$ length-2 operators from $\mL_2$. Any operator of odd-length $\ell=2i+1$, is created by appending $\Gamma_{(L-1)/2}$ to a length-$2i$ operator.
\end{enumerate}
Putting together all the operators created in Steps[1]-[3], we get the desired cardinality for the class (see \eqref{eq:countC0}), that is, $|\mC_{0}| = 2^{n} -1$. \\

\noindent\textbf{Proof of properties \textbf{(P1)} through \textbf{(P3)}:} The different length operators have again been picked in such a way as to ensure that they all commute with each other. Since the remaining $L-1$ classes are generated by successive applications of the unitary $U$ to the elements of $\mC_{0}$, we have $U: \mC_{i}\rightarrow\mC_{(i+1)\mod L}$. Thus \textbf{(P1)} and \textbf{(P3)} is satisfied. It remains to prove that the classes constructed here also satisfy property (\textbf{P2}).

As in the earlier case of $2n+1$ classes, the operators in each of the sets $\{\mL_{2}^{(0)}, \mL_2^{{(1)}},\ldots,\mL_2^{(2r-1)}\}$ correspond to unique values of the spacing function:
\begin{equation}
S(\mL_{2}^{(i)}) \equiv \{L-2,L-4,\ldots, 1\}, \forall i \in [0,2r-1], \nonumber
\end{equation}
which guarantees, by Lemma~\ref{lem:spacing} that these operators cycle through mutually disjoint sets under $U$. Since the operators in $\mL_2^{(2r)}$ are formed by combining $\Gamma$-operators from different sets $\mG^{(i)}$, each of them cycles through a different set of operators under $U$. Thus we see that all the length-$2$ operators in $\mC_0$ cycle through mutually disjoint sets. 

Before we proceed to discuss the higher length operators, we make one further observation about the length-$2$ operators. The operators in $\mL_2$ correspond to index sets which satisfy
\begin{equation}\label{eq:sumL2L}
\mL_{2}(\mC_{1}) = \{\Gamma_{i_1}\Gamma_{i_2}| \, i_1+i_2 =0\,\textrm{mod}\,L\}.  
\end{equation}
In particular, the length-2 operators in the set $\mL^{(2r)}$ have been picked carefully so as to ensure that the above constraint is satisfied. In fact, this was the rationale behind leaving out the first operator in each of the sets $\mG^{(i)}$ while choosing the corresponding length-$2$ elements in $\mL_2^{(i)}$.

The higher length operators in $\mC_0$ can be of two types:
\begin{enumerate}
\item[(a)]Those that are comprised of $\Gamma$-operators from a single set $\mG^{(i)}$ alone, and 
\item[(b)]Operators that comprise $\Gamma$-operators from more than one set. 
\end{enumerate}
Since a type-(a) operator cannot cycle into a type-(b) operator under the action of $U$, these two cases can be examined separately.\\

\bigskip

\noindent\textbf{Type-(a):} The maximum length that an operator of type-(a) can have, as per our construction, is $L-1$. We have ensured this by leaving at least one operator of each of the sets $\mG^{(i)}$ unused in constructing the length-$2$ operators. Furthermore, the constraint in \eqref{eq:sumL2L} implies that the index sets corresponding to such higher length operators in $\mC_0$, sum to the same value modulo $L$. More precisely, any even-length index set of length $\ell=2j$, where the indices are all drawn from a given set $\mG^{(i)}$, satisfies
\begin{eqnarray}\label{eq:evenL}
i_{1}+i_{2}+\ldots + i_{l} &=& 0\,\textrm{mod}\,L, \nonumber \\
&\forall& \hspace{-1mm} (i_1,i_2,\ldots,i_l) \in \mC_{0}.
\end{eqnarray}
And any index set of odd length $\ell = 2j+1$ satisfies 
\begin{eqnarray}\label{eq:oddL}
i_{1}+i_{2}+\ldots + i_{l} &=& \left(\frac{L-1}{2}\right)\,\textrm{mod}\,L, \nonumber \\
&\forall& \hspace{-1mm} (i_1,i_2,\ldots,i_l) \in \mC_{0}.
\end{eqnarray}
Then, invoking Lemma~\ref{lem:lengthl} with $c_{\ell} = 0$ for even values of $\ell$ and $c_\ell = (L-1)/2$ for odd values of $\ell$, we see that no operator of type-(a) can belong to more than one class, for prime values of $L$.\\

\bigskip

\noindent\textbf{Type-(b):} An operator of type-(b) is a product of operators from smaller sets $\mK_j \subseteq \mG^{(j)}$.
Consider a length-$\ell$ operator, $\mO$ which comprises $\ell_{0}$ $\Gamma$-operators from $\mG^{(0)}$, $\ell_{1}$ operators from $\mG^{(1)}$, and in general, $\ell_{i}$ from the set $\mG^{(i)}$. 
\[
\mO = \underbrace{\Gamma_{i_1}\ldots\Gamma_{i_{\ell_0}}}_{\mK_0\subseteq\mG^{(0)}}\underbrace{\Gamma_{j_1}\ldots\Gamma_{j_{\ell_1}}}_{\mK_1\subseteq\mG^{(1)}}\ldots\underbrace{\Gamma_{k_1}\ldots\Gamma_{k_{\ell_{2r-1}}}}_{\mK_{2r-1}\subseteq\mG^{(2r-1)}}
\]
Note that by our construction, the operator $\mO$ exists in more than one class if and only if, for all $\mK_{j}$ 
the product of all operators in $\mK_{j}$ also belongs to more than one class.
In what follows, we argue that our construction ensures that this is not possible.
In particular, given a set of length-$\ell$ operators in $\mC_0$ which can be broken down into smaller sets as described above, we will argue that there exists at least one set $\mK_{j}$ in every such length-$\ell$ operator $\mO$, such that the products of operators in $\mK_{j}$ corresponding to different length-$\ell$ operators cycle through mutually disjoint sets, as defined earlier.

Note the following two facts about the subsets $\mK_j$. First, our construction ensures that any subset $\mK_j \subseteq\mG^{(j)}$ of a given size $\ell_j$, satisfies either \eqref{eq:evenL} or \eqref{eq:oddL} depending on $\ell_j$ being even or odd. Second, note that the maximum size of these subsets 
is $\ell_{j}\leq L$. However, in order to invoke Lemma~\ref{lem:lengthl}, we still require $\ell_j$ to be strictly less than $L$. 
Our goal is hence to show that every length-$\ell$ operator must have at least one subset $\mK_j$ of size $\ell_j < L$. 

Suppose there exists a length-$\ell$ operator such that every subset is of size $L$. Then, the operator itself has to be of length
\begin{equation}
\ell = \ell_0 + \ell_1 + \ldots + \ell_{2r-1} = 2rL = 2n
\end{equation}
However the maximum value of $\ell$ in our construction is $2n-1$, implying that atleast one of the $2r$ subsets must be of a size strictly smaller than $L$. And, for such a subset of size less than $L$, constraints \eqref{eq:evenL} and \eqref{eq:oddL} ensure that the same subset cannot be found in more than one class, provided $L$ is prime. \finproof

\section{A simple lower bound on min-entropy}

The min-entropy of the distribution that an orthonormal basis $\mathcal{B}_{j} = \{\ket{b^{(j)}}\}_b$ induces on a state $\rho \in \mathcal{H}$ is given by
\begin{equation}
\mathcal{H}_{\infty}(\mathcal{B}_{j}|\rho) = -\log\max_{b}\textrm{Tr}[\ket{b^{(j)}}\bra{b^{(j)}}\rho]
\end{equation}
We are looking to evaluate a lower bound on the average min-entropy of any $L$ mutually unbiased bases (not necessarily coming from our construction) in a $d$-dimensional Hilbert space. The average min-entropy is given by -
\begin{eqnarray}\label{eq:jensen1}
\frac{1}{L}\sum_{j=0}^{L-1}\mathcal{H}_{\infty}(\mathcal{B}_{j}|\rho) &=& -\frac{1}{L}\sum_{j} \log\max_{b^{(j)} \in \{0,...,d-1\}}\langle b^{(j)}|\rho|b^{(j)}\rangle \nonumber \\
&\geq& -\log\frac{1}{L}\sum_{j=0}^{L-1}\max_{b^{(j)}}\langle b^{(j)}|\rho|b^{(j)}\rangle
\end{eqnarray}
using Jensen's inequality. The problem of finding an optimal uncertainty relation for the min-entropy, thus reduces to the problem of maximizing over all $\rho \in \mathcal{H}$, the quantity $\sum_{j=0}^{L-1}\max_{b^{(j)}\in \{0,...,d-1\}}\expect{b^{(j)}}{\rho}{b^{(j)}}$. It is easy to see that this maximum is always attained at a pure state, so we can restrict the problem to an optimization over pure states. We can simplify the problem of finding the lower bound of~\eqref{eq:jensen1} by recasting it as follows.

Consider states of the form $P_{\vec{b}} = \frac{1}{L}\sum_{j=0}^{L-1}\ket{b^{(j)}}\bra{b^{(j)}}$ where $\vec{b} = (b^{(0)}, b^{(1)},..., b^{(L-1)})$ denotes a string of basis elements, that is, $b^{(j)} \in \{0,1,...,d-1\}$. Suppose we can show for all possible strings $\vec{b}$,
\begin{equation}\label{eq:maxp}
\max_{\ket{\psi}} \textrm{Tr}(P_{\vec{b}}\ket{\psi}\bra{\psi}) \leq \zeta\ .
\end{equation}
Then, since $\frac{1}{L}\sum_{j}|\langle b^{(j)}\ket{\psi}|^{2} = \textrm{Tr}[P_{\vec{b}}\ket{\psi}\bra{\psi}]$, the bound is simply
\begin{equation}\label{eq:evbound}
\frac{1}{L}\sum_{j=0}^{L-1}\mathcal{H}_{\infty}(\mathcal{B}_{j}|\ket{\psi}\bra{\psi}) \geq -\log\zeta\ .
\end{equation}
We have thus reduced the problem to one of finding the maximum eigenvalue of operators of the form $P_{\vec{b}}$, over all possible strings $\vec{b}$.

\subsection{A new bound for smaller sets of $L<d$ MUBs}

We now prove Lemma~\ref{lem:smallLBound}, restated here for convenience.
\begin{lemma}
Let $\mB_0,\ldots,\mB_{L-1}$ be a set of mutually unbiased bases in dimension $d$. Then, 
\begin{equation}
\frac{1}{L}\sum_{j=0}^{L-1}\mathcal{H}_{\infty}(\mathcal{B}_{j}|\ket{\psi}) \geq -\log\left[\frac{1}{L}\left(1 + \frac{L-1}{\sqrt{d}}\right)\right]\ . 
\end{equation}
\end{lemma}
\begin{proof}
Note that by~\eqref{eq:evbound}, it is sufficient to determine $\zeta$ in~\eqref{eq:maxp}.
To solve this eigenvalue problem
we recall a result of Schaffner~\cite{chris:diss} proved using the methods of Kittaneh~\cite{kittaneh3:normsum}, that for a set of $L$ orthogonal projectors $A_0,A_1,\ldots,A_{L-1}$, the following bound holds:
\begin{equation}
\parallel\sum_{j=0}^{L-1}A_{j}\parallel \leq 1 + (L-1)\max_{0\leq j<k\leq L-1}\parallel A_{j}A_{k}\parallel
\end{equation}
where $\parallel (.)\parallel$ denotes the operator norm, which here is simply the maximum eigenvalue for Hermitian operators. Applying this 
result to sums of basis vectors $\ket{b^{(j)}}$, we have,
\begin{eqnarray}
&&\hspace{-1mm} \parallel \sum_{j=0}^{L-1}\ket{b^{(j)}}\bra{b^{(j)}} \parallel \; \; \leq \; \; 1 + \\
&& \; \; (L-1)\max_{0\leq j<k\leq L-1}\parallel (\ket{b^{(j)}}\bra{b^{(j)}})(\ket{b^{(k)}}\bra{b^{(k)}})\parallel \nonumber 
\end{eqnarray}
which implies
\begin{eqnarray}\label{eq:pvec}
\hspace{-2mm}&& \; \parallel P_{\vec{b}}\parallel \; \; \leq \; \; \frac{1}{L} + \\ 
&&\hspace{-1mm} \left(\frac{L-1}{L}\right)\max_{0\leq j<k\leq L-1}\parallel  \ket{b^{(j)}}(\bra{b^{(j)}}\ket{b^{(k)}})\bra{b^{(k)}}\parallel \nonumber
\end{eqnarray}
Recall, that for all $b^{(j)},b^{(k)} \in \{0,\ldots,d-1\}$
\begin{align}
\inp{b^{(j)}}{b^{(k)}} = e^{i\phi}\frac{1}{\sqrt{d}}, \, \mbox{for any} \; j \neq k,
\end{align}
where $\phi$ denotes some phase factor. Further, since the vectors $\ket{b^{(j)}}$ are normalized, the
Cauchy-Schwarz inequality gives
\begin{align}
\parallel \ket{b^{(j)}}\bra{b^{(k)}} \parallel \leq 1, \, \mbox{for any} \; b^{(j)}, b^{(k)} \in\{0,\ldots,d-1\}. \nonumber
\end{align}
Combining these with \eqref{eq:pvec} gives the following bound on the maximum eigenvalue of the operator $P_{\vec{b}}$ :
\begin{equation}
\zeta = \frac{1}{L}\left(1 + \frac{L-1}{\sqrt{d}}\right)\ .
\end{equation}
By~\eqref{eq:evbound}, this immediately proves our claim.
\end{proof}

\subsection{A stronger bound for the complete set of $d+1$ MUBs}

Here, we present an alternate approach to bound the maximum eigenvalue of $P_{\vec{b}}$, using a Bloch vector like representation of the MUB basis states. The bound that we obtain here, stated in Lemma~\ref{lem:largeLBound}, is stronger than the last one when $L>d$. In particular, when we consider the complete set ($L=d+1$) of MUBs in any dimension $d$, this approach yields the best known bound.

\begin{lemma}
Let $\mB_0,\ldots,\mB_{L-1}$ be a set of mutually unbiased bases in dimension $d$. Then,
\begin{equation}
\frac{1}{L}\sum_{j=0}^{L-1}\mathcal{H}_{\infty}(\mathcal{B}_{j}|\ket{\psi}) \geq -\log\left[\frac{1}{d}\left(1 + \frac{d-1}{\sqrt{L}}\right)\right] \ .
\end{equation}
\end{lemma}
\begin{proof}
First, we switch to working in a basis of Hermitian operators, so that every state in $\mathcal{H}$ has a parametrization in terms of vectors in a real vector space. Any state $\rho \in \mathcal{H}$ can be written as:
\begin{equation}\label{eq:Hobasis}
\rho = \frac{1}{d}\mathbb{I} + \frac{1}{2}\sum_{i=1}^{d^2-1}\alpha^{(i)}\hat{A}_{i}
\end{equation}
where $\{\hat{A}_{i}\}$ are Hermitian, trace-less operators that are orthogonal with respect to the Hilbert-Schmidt norm: $\textrm{Tr}[\hat{A}_{i}^\dagger\hat{A}_{j}] = 2\;\delta_{ij}$, and the scalars $\{\alpha^{(i)}\}_i \in \mathbb{R}$. Thus we can parameterize any state in our $d$-dimensional Hilbert space with a vector $\vec{\alpha} = (\alpha^{(1)},...,\alpha^{(d^2-1)}) \in \mathbb{R}^{d^2-1}$. When $\rho$ is a pure state ($\textrm{Tr}[\rho^{2}] = 1$), the vector $\vec{\alpha}$ corresponding to this pure state satisfies the following normalization condition
\begin{eqnarray}\label{eq:alphanorm}
\textrm{Tr}\left[\left(\frac{1}{d}\mathbb{I} + \frac{1}{2}\sum_{i=1}^{d^2-1}\alpha^{(i)}\hat{A}_{i}\right)^{2}\right] &=& 1 \nonumber \\
\Rightarrow \frac{1}{d} + \frac{1}{2}\sum_{i=1}^{d^2-1}|\alpha^{(i)}|^{2} &=& 1 \nonumber \\
\Rightarrow |\vec{\alpha}| = \sqrt{\sum_{i=1}^{d^2-1}|\alpha^{(i)}|^{2}} &=& \sqrt{\frac{2(d-1)}{d}}
\end{eqnarray}

Furthermore, in this representation, the vectors $\{\vec{\alpha}_{(b,j)}\}$ corresponding to the MUB states $\{\ket{b^{(j)}}\}$ satisfy the following special properties:- 
\begin{itemize}
\item (M1) \emph{Normalization}:  $\textrm{Tr}[\ket{b^{(j)}}\bra{b^{(j)}}\ket{b^{(j)}}\bra{b^{(j)}}] = 1$ implies that $|\vec{\alpha}_{(b,j)}| = \sqrt{\frac{2(d-1)}{d}}$, $\forall \; b \in \{0,...,d-1\} \; , \; j \in \{0,...,L-1\}$. (By an argument similar to the one that leads to \eqref{eq:alphanorm}.)
\item (M2) \emph{Constant inner-product}:  $|\langle b^{(j)}| \hat{b}^{(k)}\rangle|^{2} = \frac{1}{d}$ implies that $\vec{\alpha}_{(b,j)}.\vec{\alpha}_{(\hat{b},k)} = 0 , \; \forall \; j \, \neq \, k, \; \forall \; b,\hat{b} \; \in \; \{0,...,d-1\}$. This is easily seen, as follows:
\begin{eqnarray}\label{eq:alphaorth}
\textrm{Tr}[\ket{b^{(j)}}\bra{b^{(j)}}\ket{b^{(k)}}\bra{b^{(k)}}] &=& \frac{1}{d} + \frac{1}{2}\sum_{i}\alpha^{(i)}_{(b,j)}\alpha^{(i)}_{(\hat{b},k)} = \frac{1}{d} \nonumber \\
\Rightarrow \vec{\alpha}_{(b,j)}.\vec{\alpha}_{(\hat{b},k)} &=& 0
\end{eqnarray}
\end{itemize}

Now, using this representation of MUB states and density operators, we can rewrite the maximization problem of ~\eqref{eq:maxp} as:
\begin{eqnarray}
&&\max_{\ket{\psi}}\textrm{Tr}[P_{\vec{b}}\ket{\psi}\bra{\psi}] = \max_{\ket{\psi}}\textrm{Tr}\left[\frac{1}{L}\sum_{j}\ket{b^{(j)}}\bra{b^{(j)}}\ket{\psi}\bra{\psi}\right]\nonumber \\
&\leq&  \max_{\vec{\alpha}}\frac{1}{L}\sum_{j}\textrm{Tr}\left[\left(\frac{\mathbb{I}}{d} + \frac{\sum_{j}\alpha^{j}_{(b^{(j)},j)}\hat{A}_{j}}{2}\right)\left(\frac{\mathbb{I}}{d} + \frac{\sum_{i}\alpha^{(i)}\hat{A}_{i}}{2}\right)\right] \nonumber \\
&=& \max_{\vec{\alpha}}\frac{1}{L}\sum_{j}\left(\frac{1}{d} + \frac{1}{2}\vec{\alpha}_{(b^{(j)}, j)}.\vec{\alpha}\right) \nonumber \\
&=& \frac{1}{d} + \max_{\vec{\alpha}}\frac{1}{2L}\sum_{j}\vec{\alpha}_{(b^{(j)}, j)}.\vec{\alpha} 
\end{eqnarray}
Now we only need to find the real $(d^{2}-1)$-dimensional vector $\vec{\alpha}$, that maximizes the sum $\sum_{j}\vec{\alpha}_{(b^{(j)}, j)}.\vec{\alpha}$. If we now define an ``average'' vector corresponding to each string $\vec{b}$, as follows
\begin{equation}
\frac{1}{L}\sum_{j}\vec{\alpha}_{(b^{(j)}, j)} = \vec{\alpha}_{(\textrm{avg})}
\end{equation}
then, it becomes obvious that the maximum is attained when $\vec{\alpha}$ is parallel to $\vec{\alpha}^{(\textrm{avg})}$. Since it is a vector corresponding to a pure state, its norm is given by ~\eqref{eq:alphanorm}, so that
\begin{equation}
\vec{\alpha}_{(\textrm{max})} = \sqrt{\frac{2(d-1)}{d}}\frac{\vec{\alpha}_{(\textrm{avg})}}{|\vec{\alpha}_{(\textrm{avg})}|}
\end{equation}
Note that this maximizing vector has a constant overlap with all vectors $\vec{\alpha}^{(b^{(j)}, j)}$, for a given string $\vec{b}$. In other words, for each string $\vec{b}$, the maximum is attained by the vector that makes equal angles with all the vectors that constitute the ``average'' vector ($\vec{\alpha}_{(\textrm{avg})}$) corresponding to that string. Note however that this vector may not always correspond to a valid state.

Now that we know the maximizing vector, we can go ahead and compute the value of $\zeta$ in \eqref{eq:maxp}.
\begin{eqnarray}\label{eq:bound}
\max_{\ket{\psi}}\textrm{Tr}[P_{\vec{b}}\ket{\psi}\bra{\psi}] &\leq& \frac{1}{d} + \max_{\vec{\alpha}}\frac{1}{2L}\sum_{j}\vec{\alpha}_{(b^{(j)}, j)}.\vec{\alpha} \nonumber \\
&=& \frac{1}{d} + \frac{1}{2}\max_{\vec{\alpha}}\vec{\alpha}_{(\textrm{avg})}.\vec{\alpha} \nonumber \\
&=& \frac{1}{d} + \frac{1}{2}\frac{\vec{\alpha}_{(\textrm{avg})}.\vec{\alpha}_{(\textrm{avg})}}{|\vec{\alpha}_{(\textrm{avg})}|}\sqrt{\frac{2(d-1)}{d}} \nonumber \\
&=& \frac{1}{d} + \frac{1}{2}|\vec{\alpha}_{(\textrm{avg})}|\sqrt{\frac{2(d-1)}{d}} \nonumber \\
&=& \frac{1}{d} + \frac{1}{2\sqrt{L}}\frac{2(d-1)}{d} \nonumber \\
&=& \frac{1}{d}\left(1 + \frac{d-1}{\sqrt{L}}\right)
\end{eqnarray}
where we have used the fact that the vector $\vec{\alpha}_{(\textrm{avg})}$ have a constant norm
which can be computed as follows:
\begin{eqnarray}
\vec{\alpha}_{(\textrm{avg})}.\vec{\alpha}_{(\textrm{avg})} &=& \frac{1}{L^{2}}\sum_{j,k}\vec{\alpha}^{(b^{(k)},k)}.\vec{\alpha}^{(b^{(j)},j)} \nonumber \\
&=& \frac{1}{L^{2}}\sum_{j}\vec{\alpha}_{(b^{(j)},j)}.\vec{\alpha}_{(b^{(j)},j)} \nonumber \\
&=& \frac{1}{L^{2}}(L)\left[\frac{2(d-1)}{d}\right] \nonumber \\
\Rightarrow |\vec{\alpha}_{(\textrm{avg})}| &=& \frac{1}{\sqrt{L}}\sqrt{\frac{2(d-1)}{d}},
\end{eqnarray}
thus proving our claim. The second step follows from the fact that vectors corresponding to different MUB states have zero inner product (see property $(M2)$ above).
\end{proof}

Note that the fact that the bases are mutually unbiased was crucial in giving rise to properties $(M1)$ and $(M2)$ which in turn enabled us to identify the maximizing vector $\alpha_{\textrm{max}}$. Indeed the maximizing vector corresponding to a given string $\vec{b}$ might not always correspond to a 
valid state, in which case the bound we derive cannot be achieved. 
However, there exist strings of basis elements $\vec{b}$, for which we can explicitly construct a state that has equal trace overlap with the states that constitute the corresponding operator $P_{\vec{b}}$. 
These are in fact states of the form
\begin{equation}
P_{\vec{b}} = \frac{1}{L}\sum_{j}\ket{b^{(j)}}\bra{b^{(j)}},\, \mbox{where} \; \vec{b} = \{c,...,c\},
\end{equation}
for any $c\in \{0,\dots,d-1\}$. Clearly, for the symmetric MUBs that we construct, an eigenstate of the unitary $U$ that cycles between the different MUBs has the same trace overlap with each of the states $\{\ket{b^{(j)}}, j = 0,\ldots,L-1\}$, for a fixed value of $b$. To see this, suppose $\ket{\phi}$ is an eigenvector of $U$ with eigenvalue $\lambda$, then for all $0 \leq j \leq L-1$ and a given value of $b$,
\begin{eqnarray}
\textrm{Tr}[\ket{b^{(j)}}\bra{b^{(j)}}\ket{\phi}\bra{\phi}] &=& |\langle b^{(j)}|\phi\rangle|^{2} = |\langle b^{(1)}|(U^\dagger)^{j-1}|\phi\rangle|^{2} \nonumber \\
&=& (|\lambda|^{2})|\langle b^{(1)}|\phi\rangle|^{2} \nonumber \\
&=& |\langle b^{(1)}|\phi\rangle|^{2} \; 
\end{eqnarray}
This is indeed the case for $L=4$ MUBs in $d=4$, where the lower bound we derive is achieved by eigenstates of $U$.

\end{document}